\documentclass[10pt,fleqn]{article}

\usepackage{amsmath, amssymb, amsfonts, appendix}
\usepackage{amsthm}
\usepackage{extarrows}

\usepackage{enumerate}

\usepackage{graphicx}
\usepackage{eepic,color,colordvi,amscd}
\usepackage{subcaption}

\newtheorem{theorem}{Theorem}[section]
\newtheorem{lemma}[theorem]{Lemma}
\newtheorem{definition}[theorem]{Definition}

\newtheorem{problem}[theorem]{Problem}

\newtheorem{remark}[theorem]{Remark}
\newtheorem{corollary}[theorem]{Corollary}

\newcommand{\R}{\mathbb{R}}
\newcommand{\Z}{\mathbb{Z}}

\newcommand{\rs}{\setminus}

\newcommand{\mc}[1]{\mathcal{#1}}
\newcommand{\imp}{\Rightarrow}
\newcommand{\argmin}{\operatorname{argmin}}

\usepackage{graphicx}
\usepackage{tikz} 
\usepackage{hyperref}
\hypersetup{
	colorlinks=true,
	linkcolor=blue,
	filecolor=magenta,      
	urlcolor=cyan,
	pdftitle={},
}

\title{Eulerian-spanning set and coboundary operator: An investigation of maxcut beyond planar graphs}
\author{Qiming Fang\footnotemark[1] \and Sihong Shao\footnotemark[3] \and Yuxuan Wu\footnotemark[3]}
\date{\today}

\footnotetext[1]{Beijing International Center for Mathematical Research, Peking University, Beijing 100871, China. Email: \texttt{fangqiming@bicmr.pku.edu.cn}}
\footnotetext[2]{CAPT, LMAM and School of Mathematical Sciences, Peking University, Beijing 100871, China. Email: \texttt{sihong@math.pku.edu.cn}}
\footnotetext[3]{School of Mathematical Sciences, Peking University, Beijing 100871, China. Email: \texttt{snrhzn@stu.pku.edu.cn}}

\begin{document}
	
	\maketitle
	{\small {\textbf{Abstract}:} Using the concepts of Eulerian-spanning set and coboundary operator, we generalize Hadlock's conversion of the maxcut problem on planar graphs to one on general graphs with non-negative weights. Using our conversion, we can explore algorithms for maxcut beyond the class of planar graphs. We obtain a Fixed-Parameter Tractable algorithm for $k$-contraction apex graphs. Specifically, our algorithm can be applied to graphs with crossing number $k$, giving an $O(2^k(n+k)^{3/2}\log (n+k))$-time algorithm that matches the best known results when restricted to non-negative weights.
		
		{\textbf{Keywords:} maxcut, planar graphs, Eulerian-spanning set, coboundary operator, Fixed-Parameter Tractable algorithm.} 
	}

	\section{Introduction}\label{s1}
	The famous maxcut problem asks for a partition of vertices of a given graph into two sets such that the number of edges that cross between two parts attains maximum. It is an NP-complete problem \cite{80}, thus being unlikely to be solved efficiently in general. However, a major breakthrough by Hadlock \cite{79} showed that the maximum cut can be found in polynomial time on planar graphs. He showed that finding the maximum cut in a planar graph is equivalent to finding a minimum odd-circuit cover, which is further equivalent to finding a minimum odd-vertex pairing in the dual graph. While the problem of minimum odd-vertex pairing can be converted to a minimum matching problem, which can be solved using Edmonds' Blossom Algorithm \cite{104}, he concluded that we can find the maximum cut in a planar graph within polynomial time. \par
	In this paper, we start with Hadlock's conversion, and generalize it beyond the setting of planar graphs. The second step of Hadlock's conversion relies heavily on planar duality, whose generalization to general graphs remains to be established. With regard to this, we take a step back and see what happens. We noticed that on planar graphs, an odd-vertex pairing in the dual graph corresponds to an edge set in the original graph satisfying the following coboundary condition: its coboundary (over the field $\Z_2$) is the set of odd faces. Therefore, a minimum odd-circuit cover is simply a minimum edge set satisfying the coboundary condition. But this coboundary condition can be generalized to the setting of general graphs by replacing the concept of face with our concept of \emph{Eulerian-spanning set}. This idea rises from the fact that the boundaries of faces in a planar graph span its cycle space, i.e. the space of Eulerian subgraphs. Therefore we generalize Hadlock's conversion on planar graphs to one that applies to maxcut on general graphs. Moreover, when restricted to planar graphs, our conversion is almost identical to Hadlock's (c.f. Section \ref{s0}). Figure \ref{diagram} illustrates our generalization, and conversion of the maxcut problem.
	\paragraph{\bf Our contribution.} By noticing the coboundary condition on planar graphs, and generalizing it with the concept of Eulerian-spanning set, we obtain a conversion of maxcut on general graphs with non-negative weights. Using our conversion, the maxcut problem is eventually converted to a \emph{weighted set symmetric difference problem (WSSD)}, which is polynomial-time solvable under specific conditions. We obtain a fixed-parameter algorithm that solves maxcut on graphs with a Eulerian-spanning set satisfying a \emph{$k$-frequent appearance condition}, which is equivalent to the graph being a $k$-contraction apex graph (the concept of $k$-contraction apex graphs is a generalization of the concept of contraction apex graphs in \cite{115}, c.f. Definition \ref{def11}). This result can be modified to cover the case of graphs with crossing number $k$.\par
	\begin{figure}[h]
		\centering
		\includegraphics[scale=.6]{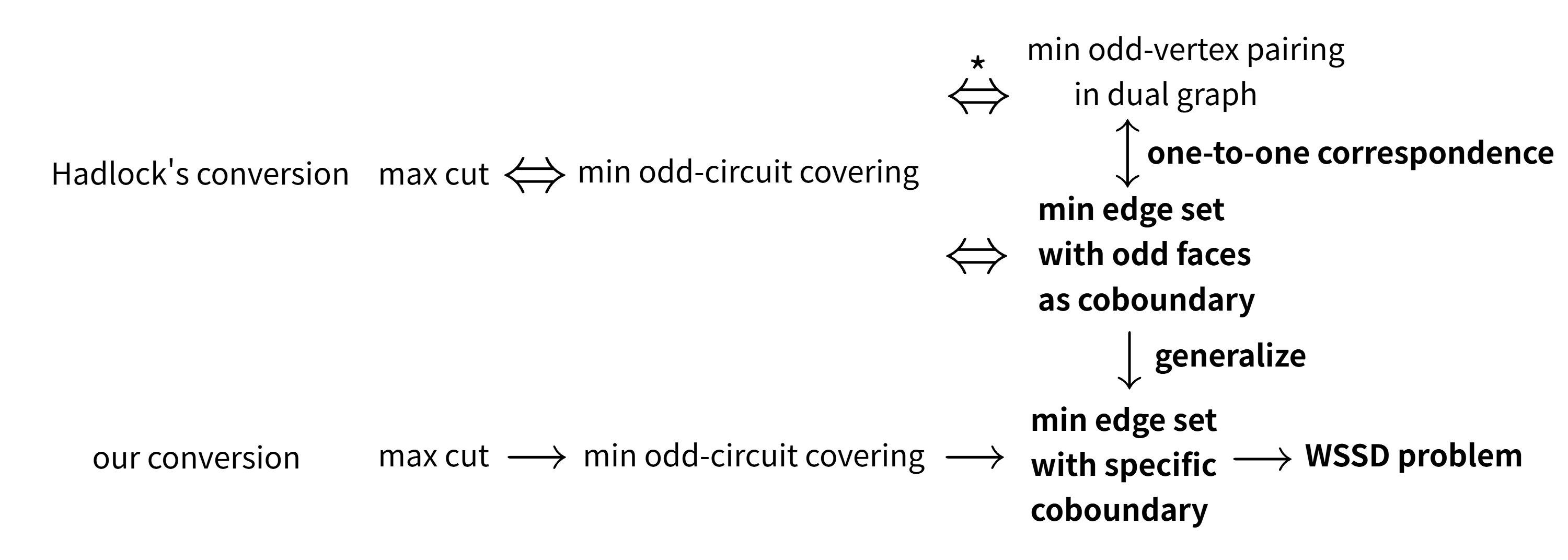}
		\caption{A comparison between Hadlock's conversion \cite{79} and ours. The star mark means that this step relies on planar duality, which is yet to be generalized to general graphs. We use bold font to emphasize what we notice, our generalization and our conversion.}
		\label{diagram}
	\end{figure}
	\vspace{10pt} \par
	\begin{definition} \label{def11}
		An undirected graph $G=(V,E)$ is a $k$-contraction apex graph, if we can contract at most $k$ edges in $G$ and get a planar graph.
	\end{definition}
	Throughout this paper, we will use $P_k(G)$ to denote an edge set of cardinality at most $k$ such that after contracting $P_k(G)$, $G$ becomes a planar graph. \par
	In this work, we give an algorithm for maxcut parameterized by the number $k$, which comes down to the following main theorem. 
	\begin{theorem}\label{thm11}
		Suppose $G=(V,E)$ is an undirected, simple, $k$-contraction apex graph, $w:E\to \R^+\cup\{0\}$ is a weight function, and $P_k(G)$ is included in the input. Let $n=|V|,m=|E|$ Then there is an $O(2^k\max\{m,n\}^{3/2}\log\max\{m,n\})$-time algorithm finding the maximum weighted cut of $G$. \par
	\end{theorem} \par
	It is easy to see that planar graphs are simply the case of $k=0$ in Theorem \ref{thm11}. Moreover, in Section \ref{s7}, we will show that an immediate corollary (Corollary \ref{col62}) of this theorem covers the case of graphs with crossing number $k$, giving an $O(2^k(n+k)^{3/2}\log(n+k))$-time algorithm (c.f. Theorem \ref{thm72}), which matches the results of \cite{116,117} restricted to the setting of non-negative weights. We note that our algorithm is essentially different from those in \cite{116,117}, which directly branch the original maxcut problem to multiple maxcut problems.
	\paragraph{\bf Related work.} Polynomial-time algorithms for max-cut have also been found in other classes of graphs, such as graphs not contractible to $K_5$ \cite{92}, graphs without long odd cycles \cite{94}, graphs of bounded treewidth \cite{106}, line graphs \cite{107}, and weakly bipartite graphs \cite{91}.
	\paragraph{\bf Organization of this paper. } In Section \ref{s0}, we illustrate the application of our conversion to planar graphs as a rough sketch of the main ideas. In Section \ref{s2} we give the formal definition of Eulerian-spanning set and coboundary operators. In Section \ref{s3} we describe the detailed conversion from the weighted maximum cut problem to a WSSD problem. In Section \ref{s5} we describe how to solve WSSD under the condition $|S_i|\leq 2,\forall S\in\mc{S}$. In Section \ref{s6} we prove Theorem \ref{thm11} by finding a Eulerian-spanning set for $k$-contraction apex graphs, and utilizing the planarity. In Section \ref{s7} we prove that the immediate corollary of Theorem \ref{thm11} (c.f. Corollary \ref{col62}) covers the case of graphs with crossing number $k$. \par
	\paragraph{\bf Basic assumptions. }All graphs considered in this paper are finite and undirected, but are allowed to have multiple edges and self-loops unless otherwise specified. When we talk about a planar graph, we will always assume that its planar embedding is arbitrarily fixed unless otherwise specified. Throughout this paper, copies of the same instance in a multiset are considered to be different.

	\section{Our conversion on planar graphs} \label{s0}
	In this section we briefly illustrate our conversion on planar graphs. The purpose is to prepare the reader for understanding our conversion on general graphs. We consider unweighted maximum cut on simple planar graphs for simplicity. We set up some definitions and notations first. 
	\begin{definition} \label{def21}
		Suppose $S$ is a finite multiset. Then $\Z_2^{S}$ is the linear space over $\Z_2$ (the field with two elements) with dimension $|S|$. Moreover, every element in $\Z_2^S$ is regarded as an indicator function of some subset of $S$.
	\end{definition} \par
	\begin{definition}
		Suppose $S_1,S_2$ are two subsets of the same multiset $S$. Then $S_1\triangle S_2$ is the symmetric difference between $S_1,S_2$.
	\end{definition}\par
	From now on, for a subset $S_1\subseteq S$, the term $S_1$ denotes both the subset, and the corresponding indicator vector in $\Z_2^S$. Therefore, $\forall S_1\subseteq S,S_2\subseteq S$, $S_1+S_2$ has the same meaning as $S_1\triangle S_2$. We use ``+" when we want to emphasize the addition with respect to $\Z_2$, and use ``$\triangle$" otherwise.\par
	\begin{definition}
		Suppose $G=(V,E)$ is a graph. Then $\partial_1:\Z_2^{E}\to \Z_2^V$ is the unique linear operator such that $\forall e\in E$, if $e=(v_1,v_2)$, then $\partial_1(\{e\})=\{v_1\}+\{v_2\}$. \par
	\end{definition}
	\begin{definition}
		Suppose $G=(V,E)$ is a planar graph. We define the boundary of a face $f$ to be the set $\{e\in E:\text{only one of the two faces on $e$'s two sides is }f\}$. We say that $f$ is odd, if the boundary of $f$ consists of an odd number of edges, and even otherwise.
	\end{definition}
	\begin{definition}[\it{coboundary operator for planar graphs}]
		Suppose $G=(V,E)$ is a planar graph. Let $F$ be the set of faces of $G$. Then $\delta_p:\Z_2^{E}\to \Z_2^{F}$ is the unique linear operator such that $\forall e\in E$, if $f_1,f_2$ are the two faces on two sides of $e$, then $\delta_p(\{e\})=\{f_1\}+\{f_2\}$.\par
	\end{definition} \par
	It is easy to see that the image of a pendant edge under $\delta_p$ is zero, because two faces on two sides of a pendant edge are the same and therefore cancel out. Now we begin converting the maxcut problem on planar graphs. First, Hadlock showed that finding the maximum cut in a planar graph is equivalent to finding a minimum odd-circuit cover:
	\begin{theorem}[\it{\cite{79},Corollary 1}]\label{thm01}
		Suppose $G=(V,E)$ is a simple graph. An edge set $E_1\subseteq E$ is a maximum cut if and only if $E\rs E_1$ is a minimum odd-circuit cover.
	\end{theorem} \par
	Now we show that a minimum odd-circuit cover is also a minimum cardinality set whose image under $\delta_p$ is the set of odd faces.
	\begin{theorem}\label{thm02}
		Suppose $G=(V,E)$ is a simple planar graph, and $F$ is the set of faces of $G$. Let $F_o$ be the set of odd faces in $G$. We have (1)$\imp$(2)$\imp$(3):\par
		(1) $E_1$ is a minimum odd-circuit cover of $G$. \par
		(2) $\delta_p(E_1) = F_o$. Or equivalently, $\forall f\in F$, if $C_f$ is the boundary of $f$, then $|C_f\cap E_1|\equiv|C_f|\pmod{2}$. \par
		(3) $E_1$ is an odd-circuit cover of $G$.
	\end{theorem}
	\begin{proof}
		(1)$\imp$(2): Suppose $C_f\cap E_1=
		\{e_1,e_2,\cdots,e_t\}$, where $t=0$ is allowed. Because $E_1$ is minimal, $\forall i\in\{1,2,\cdots,t\},E_1\rs\{e_i\}$ is not an odd-circuit cover. Therefore, there is an odd circuit $C_i$ that intersects $E_1$ only in $e_i$. Let us consider $C_{tot}=C_f+\sum_{i=1}^tC_i$. Then $C_{tot}\in C(G)$. Therefore $C_{tot}$ is the disjoint union of some circuits (it can happen that $C_{tot}=\emptyset$). Moreover, $C_{tot}\cap E_1=\emptyset$. Therefore, $|C_{tot}|$ must be even, or there is an odd circuit $C_o'\subseteq C_{tot}$ that does not intersect $E_1$, a contradiction. But $\forall i\in\{1,2,\cdots,t\},|C_i|$ is odd. We then have that the parity of $|C_f|$ must be the same as $t$, which is precisely the conclusion we want. \par
		(2)$\imp$(3): We only need to prove that every cycle in $G$ that does not intersect $E_1$ is even, because circuits can be decomposed to disjoint cycles. Suppose $C$ is a cycle in $G$ that does not intersect $E_1$. Then $C$ enclose a region in the plane. Let those faces in that region be $f_1,\cdots,f_t$, and their boundaries be $C_1,\cdots,C_t$. Then we have $C=\sum_{i=1}^tC_i$, where $\forall i,C_i\in\mc{C}$. Therefore $C\cap E_1=\sum_{i=1}^t(C_i\cap E_1)$. Then $|C|\equiv\sum_{i=1}^t|C_i|\pmod{2},|C\cap E_1|\equiv\sum_{i=1}^t|(C_i\cap E_1)|\pmod{2}$. Since $|C\cap E_1| = 0$, we have $\sum_{i=1}^t|(C_i\cap E_1)|\equiv 0\pmod{2}$. But $|C_i\cap E_1|\equiv|C_i|\pmod{2}$ by (2). Therefore $|C|\equiv\sum_{i=1}^t|C_i|\equiv\sum_{i=1}^t|(C_i\cap E_1)|\equiv 0\pmod{2}$.
	\end{proof} \par

	\begin{corollary}\label{col01}
		Suppose $G=(V,E)$ is a simple planar graph. Let $F_o$ be the set of odd faces in $G$. Then an edge set $E_1\subseteq E$ is a minimum odd-circuit cover if and only if it is a minimum cardinality set such that $\delta_p(E_1)=F_o$ .
	\end{corollary} \par
	Therefore, we need only to find a minimum cardinality set $E_1$ such that $\delta_p(E_1)=F_o$. Let the set of faces of $G$ be $F$. Then $\forall e\in E, \delta_p(\{e\})\subseteq F$. Moreover, $\forall E_s\subseteq E,\delta_{p}(E_s)=\triangle_{e\in E_s}\delta_p(e)$. Therefore, the problem can be further converted to the following \emph{set symmetric difference problem (SSD):}
	\begin{problem}[\it{SSD}]\label{prb01}
		Suppose $S$ is a set, and $\mathcal{S}=\{S_i\}_{i=1}^{m}\subseteq\Z_2^S$ is a multiset consisting of subsets of $S$. Suppose we are given a target set $T\subseteq S$. The goal is to find a minimum cardinality multiset $\mathcal{T}\subseteq\mc{S}$, such that $\triangle_{S_i\in\mc{T}}S_i=T$, or claim that no such $\mc{T}$ exists.
	\end{problem} \par
	If we let $S=F,\mc{S}=\{\delta_p(\{e\}):e\in E\},T=F_o$, then this instance of SSD is equivalent to finding a minimum cardinality set $E_1$ such that $\delta_p(E_1)=F_o$. Obviously the conversion can be done in polynomial time. Because $G$ is planar, we have $|\delta_p(\{e\}|\in\{0,2\},\forall e\in E$. But sets of size 0 can simply be discarded. Therefore we only need to solve SSD in polynomial time when $|S_i|=2,\forall S_i\in\mc{S}$. But if we view elements in $S$ as vertices, subsets $S_i\in\mc{S}$ as edges, we get the following vertex pairing problem:
	\begin{problem}[\it{vertex pairing problem}]
		Given a graph $G=(V,E)$, and a vertex set $V_t\subseteq V$, find the minimum cardinality edge set $F\subseteq E$ such that $\partial_1(F)=V_t$, or claim that no such $F$ exists.
	\end{problem}
	This vertex pairing problem can be easily solved in polynomial time by Dijkstra's shortest path algorithm and weighted perfect matching algorithm. For a complete proof, see Lemma \ref{lem51}. We note that the vertex pairing problem we get eventually is the same as the odd-vertex pairing problem as in Hadlock's work \cite{79}. \par
	Therefore, we have illustrated our conversion on planar graphs. The conversion on general graphs is similar, only that the faces are replaced by our concept of Eulerian-spanning set. The rest of the paper is devoted to proving our main results.
	\section{Eulerian-spanning set and coboundary operator}\label{s2}
	Now we describe our conversion on general graphs. We consider weighted maxcut on simple graphs, because multiple edges and self-loops can be handled easily. We set up some definitions, some of which are standard and are included only for completeness.
	\begin{definition}\label{def22}
		Given a graph $G=(V,E)$, a Eulerian subgraph of $G$ is a subgraph $G_1=(V,E_1),E_1\subseteq E$, such that every vertex $v\in V$ is of even degree in $G_1$. \par
		The cycle space of $G$, denoted by $C(G)$, is the set of the edge sets of its Eulerian subgraphs. $C(G)$ is a subspace of $\Z_2^E$.
	\end{definition} \par
	It can be easily verified that every element in $C(G)$ is the disjoint union of some cycles in $G$. \par
	From now on, we may use a subset $E_1\subseteq E$ to denote a subgraph of $G$ if the vertex set does not affect our deduction, or is clear from context.
	\begin{definition}\label{def23}
		Given a graph $G=(V,E)$, a Eulerian-spanning set of $G$ is a finite multiset $\mc{C}$ consisting of elements in $C(G)$, such that the span of elements in $\mc{C}$ is $C(G)$.
	\end{definition} \par
	By this definition, every Eulerian subgraph of $G$ is the symmetric difference between some elements in $\mc{C}$. In the setting of planar graphs, the boundary of all faces consists a Eulerian-spanning set. With regard to this, one may find that our subsequent conversion comes down to what we described in Section \ref{s0} when restricted to planar graphs.
	\begin{remark}
		We do not require elements in $\mc{C}$ to be linearly independent, because finding a maximum independent subset can be computationally costly (Gaussian elimination takes time $O(n^3)$, where $n$ is the order of the matrix). Removing the independence condition also makes it easier to prove some theorems. \par
	\end{remark}
	\begin{remark}
		We allow multiple edges and self-loops to appear in the above definition. This is because in Section \ref{s6}, we will construct Eulerian-spanning set on graphs possibly with multiple edges and self-loops.
	\end{remark}
	We then give the definition of $k$-frequent appearance condition:
	\begin{definition}\label{def24}
		Suppose $G=(V,E)$ is an undirected graph, and $\mc{C}$ is a Eulerian-spanning set of $G$. We say $\mc{C}$ satisfies the $k$-frequent appearance condition, if $|\{e\in E:|\{C_i\in\mc{C}:e\in C_i\}|\geq 3\}|\leq k$.
	\end{definition}
	\begin{definition}[\it{coboundary operators for general graphs with a Eulerian-spanning set}]\label{def25}
		Given a graph $G=(V,E)$ and a Eulerian-spanning set $\mc{C}$ of $G$, we make the following definitions:\par
		(i) $\partial_1:\Z_2^{E}\to \Z_2^V$ is the unique linear operator such that $\forall e\in E$, if $e=(v_1,v_2)$, then $\partial_1(\{e\})=\{v_1\}+\{v_2\}$. \par
		(ii) $\delta_1:\Z_2^{E}\to \Z_2^{\mc{C}}$ is the unique linear operator such that $\forall e\in E,\delta_1(\{e\})=\{C\in\mc{C}:e\in C\}$. \par
	\end{definition} \par
	By the above definition, it is obvious that
	$$\forall E_1\subseteq E,\delta_1(E_1)=\{C\in\mc{C}:|C\cap E_1|\equiv 1\pmod{2}\}$$ \par
	We additionally note that we do not need a Eulerian-spanning set to define the operator $\partial_1$.

	\section{Converting weighted maxcut to WSSD}\label{s3}
	The first step is similar to that in Section \ref{s0}.
	\begin{theorem}[\it{\cite{79},Theorem 1}]\label{thm30}
		Suppose $G=(V,E)$ is a simple graph. An edge set $E_1\subseteq E$ is a subset of a cut if and only if $E\rs E_1$ is an odd-circuit cover.
	\end{theorem} \par
	\begin{theorem}\label{thm31}
		Suppose $G=(V,E)$ is a simple graph, $w:E\to \R^+\cup\{0\}$ is a weight function. Let $n=|V|,m=|E|$. Suppose we are given a set $E_1\subseteq E$ such that $E\rs E_1$ is a minimum weighted odd-circuit cover, then we can find a maximum weighted cut in time $O(m+n)$.
	\end{theorem} \par
	\begin{proof}
		Because $E\rs E_1$ is an odd-circuit cover, $(V,E_1)$ is bipartite. Fix any partition $V=V_1\cup V_2,V_1\cap V_2=\emptyset$ such that all edges in $E_1$ have one endpoint in $V_1$ and the other in $V_2$. Then the partition $(V_1,V_2)$ induces a cut $E_3$ in $G$. Moreover, $E_3\supseteq E_1$. Now we prove that $E_3$ is a maximum weighted cut of $G$. Suppose $E_4$ is a cut in $G$ with total weight greater than $E_3$. Then by Theorem \ref{thm30}, $E\rs E_4$ is an odd-circuit cover with total weight smaller than $E\rs E_3$. But $E\rs E_3$ has total weight no more than $E\rs E_1$. Then $E\rs E_4$ is an odd-circuit cover with total weight smaller than $E\rs E_1$, a contradiction. Therefore $E_3$ is a maximum weighted cut, and can obviously be found in time $O(m+n)$ if $E_1$ is given.
	\end{proof} \par
	By Theorem \ref{thm31}, we only need to find a minimum weighted odd-circuit cover of $G$. \par
	Given a graph $G=(V,E)$, we distinguish elements in the cycle space $C(G)$ by a measure of parity to be described below, which is similar to the way we distinguish odd and even faces in the setting of planar graphs:
	\begin{definition}
		Suppose $G=(V,E)$ is a graph, and $C\in C(G)$. We say that $C$ is odd, if $C$ has an odd number of edges. Otherwise we say that $C$ is even.
	\end{definition} \par
	Now we prove the analogues of Theorem \ref{thm02} and Corollary \ref{col01} for Eulerian-spanning set on general graphs.
	\begin{theorem}\label{thm32}
		Suppose $G=(V,E)$ is a simple graph, and $\mc{C}=\{C_i\}_{i\in[m]}$ is a Eulerian-spanning set of $G$. Let $\mc{C}_o\subseteq\mc{C}$ be the multiset of odd elements in $\mc{C}$. Suppose $E_1\subseteq E$. We have (1)$\imp$(2)$\imp$(3):\par
		(1) $E_1$ is a minimal odd-circuit cover of $G$ with respect to set inclusion. That is, $E_1$ itself is an odd-circuit cover, but $\forall E_2\subsetneq E_1$, $E_2$ is not an odd-circuit cover. \par
		(2) $\delta_1(E_1) = \mc{C}_o$. Or equivalently, $\forall C\in \mathcal{C},|C\cap E_1|\equiv|C|\pmod{2}$. \par
		(3) $E_1$ is an odd-circuit cover of $G$
	\end{theorem}
	\begin{proof}
		(1)$\imp$(2): Suppose $C\cap E_1=
		\{e_1,e_2,\cdots,e_t\}$, where $t=0$ is allowed. Because $E_1$ is minimal with respect to set inclusion, $\forall i\in\{1,2,\cdots,t\},E_1\rs\{e_i\}$ is not an odd-circuit cover. That is, $\forall i\in\{1,2,\cdots,t\}$, there is an odd circuit $C_i$ that intersects $E_1$ only in $e_i$. Now let us consider $C_{tot}=C+\sum_{i=1}^tC_i$. Then $C_{tot}\in C(G)$. Therefore $C_{tot}$ is the disjoint union of some circuits (it can happen that $C_{tot}=\emptyset$). Moreover, $C_{tot}\cap E_1=\emptyset$. Therefore, $|C_{tot}|$ must be even, or there is an odd circuit $C_o'\subseteq C_{tot}$ that does not intersect $E_1$, a contradiction. But $\forall i\in\{1,2,\cdots,t\},|C_i|$ is odd. We then have that the parity of $|C|$ must be the same as $t$, which is precisely the conclusion we want. \par
		(2)$\imp$(3): Suppose $C$ is a circuit in $G$ that does not intersect $E_1$. By definition of $\mathcal{C}$, we have $C=\sum_{i=1}^tC_i$, where $\forall i,C_i\in\mc{C}$. Therefore $C\cap E_1=\sum_{i=1}^t(C_i\cap E_1)$. Then $|C|\equiv\sum_{i=1}^t|C_i|\pmod{2},|C\cap E_1|\equiv\sum_{i=1}^t|(C_i\cap E_1)|\pmod{2}$. Since $|C\cap E_1| = 0$, we have $\sum_{i=1}^t|(C_i\cap E_1)|\equiv 0\pmod{2}$. But $|C_i\cap E_1|\equiv|C_i|\pmod{2}$ by (2). Therefore $|C|\equiv\sum_{i=1}^t|C_i|\equiv\sum_{i=1}^t|(C_i\cap E_1)|\equiv 0\pmod{2}$.
	\end{proof} \par
	\begin{theorem}\label{thm33}
		Suppose $G=(V,E)$ is a simple graph, $w:E\to \R^+\cup\{0\}$ is a weight function, and $\mc{C}=\{C_i\}_{i\in[m]}$ is a Eulerian-spanning set of $G$. Let $\mc{C}_o\subseteq\mc{C}$ be the set of odd elements in $\mc{C}$. Let $w_{t}:\Z_2^{E}\to\R^+\cup\{0\},\forall E_s\subseteq E,w_t(E_s)=\sum_{e\in E_s}w(e)$. Suppose $E_1\subseteq E$. Then $E_1$ is a minimum weighted odd-circuit cover of $G$ if and only if 
		$$\exists E_3\subseteq E_1,w_t(E_3)=w_t(E_1),E_3=\argmin_{F\subseteq E:\delta_1(F)=\mc{C}_o}w_t(F)$$
	\end{theorem}
	\begin{proof}
		$(\imp)$: Suppose $E_1$ is a minimum weighted odd-circuit cover. By possibly removing some edges of zero weight in $E_1$, we get a set $E_3$ that is a minimal odd-circuit cover with respect to set inclusion. By Theorem \ref{thm32}, $\delta_1(E_3)=\mc{C}_o$. Moreover, if there is an $E_4$ such that $\delta_1(E_4)=\mc{C}_o,w_t(E_4)<w_t(E_3)$, then by Theorem \ref{thm32}, $E_4$ is an odd-circuit cover with total weight less than $E_3$, a contradiction. \par
		$(\Leftarrow)$:Suppose $E_3\subseteq E_1$ satisfies
		$$w_t(E_3)=w_t(E_1),E_3=\argmin_{F\subseteq E:\delta_1(F)=\mc{C}_o}w_t(F)$$ \par
		Then $E_3$ is an odd-circuit cover by Theorem \ref{thm32}. Moreover, if there is an odd-circuit cover $E_4$ with total weight less than $E_3$, there must be an edge set $E_5\subseteq E_4$ that is a minimal odd-circuit cover with respect to set inclusion. Therefore $\delta_1(E_5)=\mc{C}_o$ by Theorem \ref{thm32}. But $E_5$ has total weight less than $E_3$, which is a contradiction.
	\end{proof} \par
	The next step is also similar to that in Section \ref{s0}. For each edge $e\in E$, $\delta_1(e)$ is a subset of $\mathcal{C}$. By Theorem \ref{thm32}, it suffices to find an $E_3$ with minimum weight such that $\delta_1(E_3)=\mc{C}_o$, because $E_3$ itself would be a minimum weighted odd-circuit cover. Since $\delta_{1}(E_3)=\triangle_{e\in E_3}\delta_1(e)$, the problem of finding a minimum odd-circuit cover is converted to the following "weighted set symmetric difference" problem, or WSSD:
	\begin{problem}[\it{WSSD}]\label{prb31}
		Suppose $S$ is a finite multiset, and $\mathcal{S}=\{S_i\}_{i=1}^{m}\subseteq\Z_2^S$ is a multiset consisting of subsets of $S$. Moreover, there is a weight function $w:\mc{S}\to\R^+\cup\{0\}$. Suppose we are given a target set $T\subseteq S$. The goal is to find a multiset $\mathcal{T}\subseteq\mc{S}$, such that $\triangle_{S_i\in\mc{T}}S_i=T$, and the total weight of sets in $\mc{T}$ is minimal, or claim that no such $\mc{T}$ exists.
	\end{problem} \par
	The decision version of Problem \ref{prb31} is as follows:
	\begin{problem}[\it{decision version of Problem \ref{prb31}}]\label{prb32}
		Suppose $S$ is a finite multiset, and $\mathcal{S}=\{S_i\}_{i=1}^{m}\subseteq\Z_2^S$ is a multiset consisting of subsets of $S$. Moreover, there is weight function $w:\mc{S}\to\R^+\cup\{0\}$. Suppose we are given a target set $T\subseteq S$, and a threshold real number $r$. The goal is to decide, whether there is a multiset $\mathcal{T}\subseteq\mc{S}$, such that $\triangle_{S_i\in\mc{T}}S_i=T$, and the total weight of sets in $\mc{T}$ is no more than $r$.
	\end{problem}
	\begin{remark}
		We will prove in Appendix \ref{a1} that Problem \ref{prb32} is NP-complete, even under strong restrictions.
	\end{remark} \par
	Now we show that we can convert a weighted maxcut problem to an instance of WSSD problem efficiently:
	\begin{theorem}\label{thm34}
		Suppose we are given a simple graph $G=(V,E)$, a weight function $w_G:E\to\R^+\cup\{0\}$, and a Eulerian-spanning set $\mc{C}$ of $G$. Let $n=|V|,m=|E|$. Then the problem of finding the minimum weighted odd-circuit cover of $G$ can be converted to an instance of Problem \ref{prb31} in time $O(m|\mc{C}|+nm)$. Moreover, if we have the answer of the converted problem, we can recover the solution of the minimum weighted odd-circuit cover in time $O(m|\mc{C}|)$ 
	\end{theorem}
	\begin{proof}
		Pick any bijection $h:\mc{C}\to[|\mc{C}|]$. Let $S=[|\mc{C}|]$. For every $e\in E$, a copy of $h(\delta_1(e))$ is included in $\mc{S}$, and weighted by $w_G(e)$. Let $T=h(\mc{C}_o)$. Then we get an instance of Problem \ref{prb31}.\par
		Moreover, we mark each set $h(\delta_1(e))$ with the edge $e$ corresponding to it during our transformation. \par
		Now we consider the computational complexity. Because each element in $\mc{C}$ has length at most $m$, the reductions described above can be done in time $O(m|\mc{C}|+nm)$. Moreover, since we keep track of the edges that correspond to sets in $\mc{S}$, we can recover the minimum weighted odd-circuit cover from the solution of the converted problem in time $O(m|\mc{C}|)$. 
	\end{proof}
	
	\section{WSSD is P when $\forall S_i,|S_i|\leq2$}\label{s5}
	When $\forall S_i,|S_i|\leq2$, we can convert the WSSD problem to a weighted vertex pairing problem as in Section \ref{s0}, and solve it in polynomial time. For completeness, we first present a polynomial-time algorithm solving the weighted vertex pairing problem.
	\begin{lemma}\label{lem51}
		Suppose we are given a graph $G=(V,E)$, a weight function $w:E\to\R^+\cup\{0\}$, and a target vertex set $V_1\subseteq V$. Let $n=|V|,m=|E|$. Then there is an $O(n^2m+n^3)$ algorithm finding the edge set $F\subseteq E$ of minimum total weight such that $\partial_1(F)=V_1$, or claiming that no such $F$ exists.
	\end{lemma}
	\begin{proof}
		Because $\partial_1(F)=V_1$, the subgraph $G_1=(V,F)$ can be decomposed into disjoint 2-paths connecting vertices in $V_1$, and cycles. We can remove all cycles, because the image of a cycle under $\partial_1$ is zero, and $F$ is of minimum total weight. Therefore, to find $F$, we only need to find the shortest path connecting every pair of vertices in $V_1$ (there can be no such paths), and then compute a minimum weighted perfect matching. The procedure is described formally as follows: \par
		We construct a graph $G_{path}=(V_1,E_{path})$, where $(v_i,v_j)\in E_{path}$ if and only if there is a path connecting $v_i$ and $v_j$, and we weight $(v_i,v_j)$ by the minimum total weight among all paths connecting $v_i$ and $v_j$. Then we compute the minimum weighted perfect matching of $G_{path}$. The symmetric difference of paths corresponding to edges in the perfect matching is the $F$ we want (we note that the paths can only intersect in edges with zero weight, otherwise the symmetric difference of these paths has boundary $V_1$, and therefore corresponds to a perfect matching in $G_{path}$ with total weight even smaller than that of the minimum perfect matching we just computed, a contradiction). \par
		Now we consider the computational complexity. The shortest path problem can be solved by Dijkstra's Algorithm in time $O(n^2m)$. Deciding whether there is a perfect matching in $G_{path}$, and computing the minimum weighted one when there is can both be done in time $O(n^3)$ by the weighted matching algorithm in \cite{110}. Therefore the total running time is $O(n^2m+n^3)$.
	\end{proof}
	\begin{theorem}\label{thm51}
		Suppose $S$ is a finite multiset, $|S|=n$, and $\mathcal{S}=\{S_i\}_{i=1}^{m}\subseteq\Z_2^S$ is a multiset consisting of subsets of $S$. Moreover, every set in $\mc{S}$ has size at most 2. We are also given a weight function $w:\mc{S}\to\R^+\cup\{0\}$. Suppose there is a target set $T\subseteq S$. Then there is an $O(n^2m+n^3)$-time algorithm finding $\mathcal{T}\subseteq\mc{S}$ such that $\triangle_{S_i\in\mc{T}}S_i=T$ and the total weight of sets in $\mc{T}$ is minimum, or claiming that no such $\mc{T}$ exists.
	\end{theorem}
	\begin{proof}
		Obviously we can ignore the sets of size 0, i.e. empty sets. Therefore we assume that all sets in $\mc{S}$ are of size 1 or 2. Let $\mc{S}_1=\{S_i\in\mc{S}:|S_i|=1\},\mc{S}_2=\{S_i\in\mc{S}:|S_i|=2\}$. We construct a graph $G_s=(V_s,E_s)$ with a weight function $w_s:E_s\to\R^+\cup\{0\}$ as follows:
		$$V_s=S\cup\{u\},E_s=E_A\cup E_B$$
		$$E_A=\{(s_i,s_j):\{s_i,s_j\}\in\mc{S}_2\}$$
		$$E_B=\{(s_i,u):\{s_i\}\in\mc{S}_1\}$$
		$$\forall(s_i,s_j)\in E_A,w_s(s_i,s_j)=w(\{s_i,s_j\})$$
		$$\forall(s_i,u)\in E_B,w_s(s_i,u)=w(\{s_i\})$$ \par
		Then there is a one-to-one correspondence between sets in $\mc{S}$ and edges in $G_s$. For every $\mathcal{T}\subseteq\mc{S}$, denote its corresponding edge set by $E_{\mc{T}}$. Then the total weight of edges in $E_{\mc{T}}$ is equal to the total weight of sets in $\mc{T}$. Moreover, $\triangle_{S_i\in \mc{T}}S_i=T$ if and only if $\partial_1(E_{\mc{T}})\in\{T,T\cup\{u\}\}$. Therefore, it suffices to find an edge set $F\subseteq E_s$ of minimum total weight such that $\partial_1\in\{T,T\cup\{u\}\}$. By Lemma \ref{lem51}, this can be done in time $O((n+1)^2m+(n+1)^3)=O(n^2m+n^3)$.
	\end{proof}
	\section{Reducing the running time by planarity} \label{s6}
	We first give a fixed-parameter polynomial algorithm for Eulerian-spanning set $\mc{C}$ satisfying the $k$-frequent appearance condition.
	\begin{theorem}\label{thm52}
		Suppose $G=(V,E)$ is a simple graph, $w:E\to\R^+\cup\{0\}$ is a weight function, and $\mc{C}=\{C_i\}_{i\in I}$ is a Eulerian-spanning set of $G$ included as input. Assume the number of edges that appear in at least three elements in $\mc{C}$ is $k$. Let $n=|V|,m=|E|$. Then there is an algorithm finding the maximum weighted cut of $G$ that runs in time $O(2^k(\max\{|\mc{C}|,n\}^2m+\max\{|\mc{C}|,n\}^3)$.
	\end{theorem}
	\begin{proof}
		By Theorem \ref{thm31}, we only need to find the minimum weighted odd-circuit cover of $G$. By Theorem \ref{thm33},\ref{thm34}, it suffices to solve an instance $\mc{I}$ of Problem \ref{prb31} as follows: Pick any bijection $h:\mc{C}\to[|\mc{C}|]$. Let $S=[|\mc{C}|]$. For every $e\in E$, a copy of $h(\delta_1(e))$ is included in $\mc{S}$, and weighted by $w(e)$. Let $T=h(\mc{C}_o)$. Let $\mc{S}_l$ be the multiset consisting of copies of $h(\delta_1(e))$ for every $|\delta_1(e)|\geq 3$. then $|\mc{S}_l|=k$ by our assumption. Our algorithm to solve the problem is as follows: for every $\mc{R}\subseteq\mc{S}_l$, we find a $\mc{T}_1\subseteq\mc{S}\rs\mc{S}_l$ with minimum total weight such that $\triangle_{S_i\in\mc{T}_1}S_i=T\triangle(\triangle_{S_i\in\mc{R}}S_i)$, or claim that no such $\mc{T}_1$ exists. But the above problem satisfies the condition in Theorem \ref{thm51}, and therefore can be solved in time $O(\max\{|\mc{C}|,n\}^2m+\max\{|\mc{C}|,n\}^3)$. It is obvious that $\mc{R}\cup\mc{T}_1$ is a candidate solution of $\mc{I}$. Then, among our candidate solutions, we choose the one with minimum total weight, which must be the solution to $\mc{I}$. \par
		We then consider the computational complexity of the whole procedure. Converting the problem takes time $O(m|\mc{C}|+nm)$. Finding all candidate solutions takes time $O(2^k(\max\{|\mc{C}|,n\}^2m+\max\{|\mc{C}|,n\}^3)$, and choosing the best one takes at most the same quantity of time. Recovering the minimum weighted odd-circuit cover takes time $O(m|\mc{C}|)$, and recovering the maximum weighted cut from the minimum weighted odd-circuit cover takes time $O(m)$. Therefore, the total running time is $O(2^k(\max\{|\mc{C}|,n\}^2m+\max\{|\mc{C}|,n\}^3)$.
	\end{proof} \par
	We will prove that we can find a Eulerian-spanning set satisfying $k$-frequent appearance condition for $k$-contraction apex graphs efficiently later. Despite this, we cannot apply the above theorem in a black-box way to get Theorem \ref{thm11}, because the running time is far from what we desired. In fact, we have to utilize the planarity of $G/P_k(G)$ to cut down the running time. The remaining of this section is devoted to proving Theorem \ref{thm11}. \par
	We present definitions and set up notations first. \par
	\begin{definition}[\it{edge contraction for graphs}]\label{def51}
		Suppose $G=(V,E)$ is a graph. For $E_1\subseteq E$, we denote the resulting graph after contracting all edges in $E_1$ while preserving multiple edges and self-loops by $G/ E_1$. When $E_1=\{e\}$ is a set of a single edge, we use $G/e$ as a shorthand for $G/\{e\}$.
	\end{definition}
	Therefore the contraction of an edge set $E_1$ can be viewed as contracting edges in $E_1$ one by one. \par
	We need a definition of edge contraction for subgraphs.
	\begin{definition}[\it{edge contraction for subgraphs}]
		Suppose $G=(V,E)$ is a graph, and $H=(V_2,E_2)$ is a subgraph of $G$. Let $e\in E,e=(v_1,v_2)$. If $e\in E_2$, the contraction $H/e$ is defined exactly the same way as in Definition \ref{def51}. When $e\notin E_2$, the contraction is defined as follows: $H/e$ is the graph obtained by seeing $v_1$ and $v_2$ as equal, while preserving multiple edges and self-loops. The contraction of a subgraph with respect to an edge set $E_1\subseteq E$, which we will denote by $H/E_1$, is just the graph obtained by contracting edges in $E_1$ one by one.
	\end{definition}
	\begin{remark}
		One can easily verify that $H/E_1$ is well defined. Moreover, for any subgraph $H$ and edge set $E_1\subseteq E$, $H/E_1$ can be viewed as a subgraph of $G/E_1$.
	\end{remark} \par
	We still need a lemma before we explain how to find a Eulerian-spanning set for $k$-contraction apex graphs.
	\begin{lemma}\label{lem61}
		Suppose $G=(V,E)$ is a graph, $e\in E$, and $\mc{C}$ is a Eulerian-spanning set of $G$. Then $\mc{C}/e\triangleq\{C_i/e: C_i\in\mc{C},C_i/e\text{ has at least one edge}\}$ is a Eulerian-spanning set of $G/ e$. Moreover, for any edge $e_1$ in $G/ e$, the number of elements in $\mc{C}/ e$ containing it is no more than the number of elements in $\mc{C}$ containing its corresponding edge in $G$.
	\end{lemma}
	\begin{proof}
		Obviously $\mc{C}/ e\subseteq C(G/ e)$. Suppose that $C\in C(G/ e)$, then $\exists C'\in C(G), C'/ e=C$. If $C'$ is the symmetric difference between $C_1,C_2,\cdots,C_l\in\mc{C}$, then $C$ is obviously the symmetric difference between $C_1/ e,C_2/ e,\cdots,C_l/ e$. The latter statement is obvious as well since we do not merge multiple edges or remove self-loops.
	\end{proof}
	\begin{theorem}\label{thm61}
		Suppose $G=(V,E)$ is a graph. Let $n=|V|,m=|E|$. Then the following two statements are equivalent:\par
		(1) There is a Eulerian-spanning set $\mc{C}$ of $G$ such that the number of edges that appear in at least three elements in $\mc{C}$ is at most $k$.\par
		(2) $G$ is $k$-contraction apex. \par
		Moreover, if $G$ is simple, we have: \par
		(I) If condition (2) is satisfied and $P_k(G)$ is given, we can compute a Eulerian-spanning set for $G$ which has size $O(m)$ and satisfies condition (1) in time $O(\max\{m,n\}(k+1)^2)$. \par
		(II) If condition (1) is satisfied, we can find an edge set $P_k(G)$ within time $O(m|\mc{C}|)$.
	\end{theorem}
	\begin{proof}
		(1)$\imp$(2): We set $P_k(G)$ to be the set of edges that appear in at least three elements in $C$, and consider $G/ P_k(G)$. By repeatedly applying Lemma \ref{lem61}, we get that $G/ P_k(G)$ has a Eulerian-spanning set $\mc{C}_1$ such that each edge in $G/ P_k(G)$ is contained in at most two elements in $\mc{C}_1$, because edges in $P_k(G)$ are no longer in $G/ P_k(G)$. We then transform $\mc{C}_1$ to a cycle basis such that each edge in $G/ P_k(G)$ is contained in at most two cycles in $\mc{C}_1$, and by the Mac Lane Planarity Criterion conclude the proof. \par
		Suppose there is a n element $C_i\in \mc{C}_1$ such that $C_i$ is not a cycle. Then $C_i$ is the disjoint union of cycles. Take any cycle $C_{i1}\in C_i$. We distinguish between three cases: \par
		(i) If $C_{i1}$ cannot be represented by the symmetric difference of elements in $\mc{C}_1\rs\{C_i\}$, we replace $C_i$ by $C_{i1}$ in $\mc{C}_1$. \par
		(ii) If (i) doesn't hold, and $C_i\rs C_{i1}$ cannot be represented by the symmetric difference of elements in $\mc{C}_1\rs\{C_i\}$, we replace $C_i$ by $C_i\rs C_{i1}$ in $\mc{C}_1$. \par
		(iii) If neither (i) nor (ii) holds, we remove $C_i$ from $\mc{C}_1$. \par
		By basic linear algebra, after the modification, $\mc{C}_1$ will still be a Eulerian-spanning set of $G/ P_k(G)$. Moreover, after the modification, the total size of elements in $\mc{C}_1$ decreases, and the number of appearance of any edge in $\mc{C}_1$ cannot increase. If we repeat the above procedure until we cannot do it anymore, we will finally get a Eulerian-spanning set consisting of cycles, such that any edge cannot appear in more than two cycles in the basis. After picking a maximum independent subset, we get the cycle basis we want. \par
		Moreover, since we can find $P_k(G)$ by simply counting the number of appearance of any edge in $\mc{C}$, which can be done in time $O(m|\mc{C}|)$, we prove (II). \par
		(2)$\imp$(1): We construct a Eulerian-spanning set $\mc{C}$ for $G$ such that the number of edges that appear in at least three elements in $\mc{C}$ is at most $k$. $\mc{C}$ consists of two parts: \par
		(i) Partition $V$ into subsets corresponding to the connected components in $H=(V,P_k(G))$, and let them be $V_1^*,\cdots,V_t^*\subseteq V$. For any $i=1,2,\cdots,t$, let $P_i\triangleq G[V_i^*]=(V_i^*,E_i^*)$. Take a spanning tree $T_i$ of $P_i$ such that $T_i\subseteq P_k(G)$. Let $E_T=\cup_{i=1}^tT_i$, then $E_T\subseteq P_k(G)$. For any edge $e\in E_i^*\rs T_i$, take the unique cycle $C_e$ in $T_i\cup e$. Let the set of all cycles taken in this way be $\mc{C}_1^*$. \par
		Now we consider the computational complexity of this part when $G$ is simple. Constructing spanning trees can be done in time $O(nk)$ by a simple DFS search. Moreover, the number of all $C_e$ cannot exceed the number of possible edges in all connected components, which is $O((k+1)^2)$. Because each cycle is of length at most $n$, recording each cycle $C_e$ and adding them to $\mc{C}_1^*$ takes time $O(n(k+1)^2)$. Therefore, this part can be done in time $O(n(k+1)^2)$.\par
		(ii) Let the graph obtained by removing all self-loops in $G/ P_k(G)$ be $G_p$. Since $G/ P_k(G)$ is planar, $G_p$ is also planar. Let the graph obtained by merging all multiple edges in $G_p$ be $G_p'$, then $G_p'$ is a simple planar graph. We first compute a planar embedding of $G_p'$. Then we compute a planar embedding of $G_p$ by recovering multiple edges.\par
		The boundaries of all faces in $G_p$ form a Eulerian-spanning set for $G_p$, which we will denote by $\mc{C}_p$. Then every edge in $G_p$ appears in no more than two elements in $\mc{C}_p$. Let $\mc{C}_p=\{C_1,\cdots,C_l\}$. \par
		Take an $i\in\{1,2,\cdots,l\}$, then $C_i$ is the disjoint union of cycles. Let those cycles be $C_i^{(1)},C_i^{(2)},\cdots,C_i^{(o)}$. $\forall j=1,2,\cdots,o$, let $C_i^{(j)}=(e_{1},e_{2},\cdots,e_{h})$ with edges in cyclic order. We note that each vertex in $G_p$ corresponds to a connected component in $H$. $\forall x=1,2,\cdots,h$, $e_x$ corresponds to an edge in $G$, which we will denote by $e_{x*}$. Moreover, the two endpoints of $e_{x*}$ are in different connected components of $H$. We construct a cycle $C_i^{(j)*}$ in $G$ as follows: \par 
		Take an endpoint $u$ shared by $e_{1},e_{2}$, and let $v$ be the other endpoint of $e_{1}$. Let $v^*$ be the endpoint of $e_{1*}$ corresponding to $v$.\par
		Start from $v^*$.\par
		Walk along $e_{1*}$ to the other endpoint.\par 
		Walk along $E_T$ until we find an endpoint of $e_{2*}$. \par
		Walk along $e_{2*}$ to the other endpoint. \par
		... \par
		Walk along $e_{x*}$ to the other endpoint.\par 
		Walk along $E_T$ until we find an endpoint of $e_{(x+1)*}$. \par
		Walk along $e_{(x+1)*}$ to the other endpoint. \par
		... \par
		Walk along $e_{h*}$ to the other endpoint.\par 
		Walk along $E_T$ until we find an endpoint of $e_{1*}$, which must be $v^*$. \par
		Let $C_i^{(j)*}$ be the cycle we get after the above procedure. Let $C_i^*$ be the symmetric difference between $C_i^{(1)*},C_i^{(2)*},\cdots,C_i^{(o)*}$. Let all the $C_i^*$ constructed in this fashion form a set $\mc{C}_2^*$. \par
		Let $\mc{C}=\mc{C}_1^*\cup \mc{C}_2^*$. \par
		We first show that only edges in $E_T$ can appear in more than two elements of $\mc{C}$. For any edge $e\notin E_T$, if the endpoints of $e$ are in the same connected component of $H$, then $e$ can only appear in elements of $\mc{C}_1^*$. But by the construction, $e$ actually appears in exactly one element, namely $C_e$, of $\mc{C}_1^*$. If the endpoints of $e$ are in different connected components of $H$, then it has a unique corresponding edge in $G_p$. Since each edge in $G_p$ appears in at most two elements in $\mc{C}_p$, we have that $e$ can appear in at most two elements of $\mc{C}_2^*$. But $e$ cannot possibly appear in elements of $\mc{C}_1^*$. Therefore $e$ appears in at most two elements of $\mc{C}$. \par
		Now we will show that $\mc{C}$ is indeed a Eulerian-spanning set. Let $C^*\in C(G)$, and $E^*\subseteq C^*$ be the set of edges that cross between different connected components of $H$. Let the edges corresponding to $E^*$ in $G_p$ be $E_C$. Then we have $E_C\in C(G_p)$. Therefore $E_C$ is the symmetric difference of elements $C_1,C_2,\cdots,C_s\in\mc{C}_p$. Consider $C_{res}^*\triangleq C^*+\sum_{m=1}^sC_m^*$. Obviously $C_{res}\in C(G)$. Moreover, $C_{res}$ does not include any edges crossing between different connected components of $H$. Therefore $C_{res}$ is the disjoint union of edge sets $C_{res}^{(1)},\cdots,C_{res}^{(t)}$, where $\forall i=1,2,\cdots,t$, edges in $C_{res}^{(i)}$ have endpoints only in $V_i^{*}$. Therefore $C_{res}^{(i)}\in C(G[V_i^{*}])$, which indicates that $C_{res}^{(i)}$ is the symmetric difference between some elements in $\mc{C}_1^*$. Then $C_{res}$ must be the symmetric difference between some elements in $\mc{C}_1^*$, which indicates that $C^*$ is the symmetric difference of elements in $\mc{C}$. \par
		Now we consider the computational complexity of this part when $G$ is simple. Computing a planar embedding of $G_p'$ can be done in time $O(n)$ by applying the algorithm in \cite{109}. Recovering multiple edges to get planar embedding of $G_p$ takes time $O(m)$, because the total number of edges in $G_p$ cannot exceed that of $G$. Sweeping and recording all faces of $G_p$ takes time at most $O(m)$, because every edge in $G_p$ is on the boundary of at most 2 faces. Because each edge in $G_p$ can appear in at most two elements in $\mc{C}_p$, the total size of elements in $\mc{C}_p$ cannot exceed $2m$. In the construction of $\mc{C}_2^*$, we decomposed elements in $\mc{C}_p$ into disjoint cycles, i.e. $C_i^{(j)}$. Each edge in $C_p$ can still appear in at most two of $C_i^{(j)}$, indicating that edges not in $E_T$ can appear in at most two of $C_i^{(j)*}$. Moreover, the total number of these cycles cannot exceed $2m$. Therefore, the total length of all $C_i^{(j)}$ is at most $2mk+2m$, indicating that computing them costs time $O(km+m)$. Computing all the $C_i^*$ and adding them to $\mc{C}_2^*$ takes at most the same amount of time. Therefore, computing $\mc{C}_2^*$ costs time $O(km+m)$. \par
		Finally, the total running time is the sum of that in (i) and (ii), which is $O((k+1)^2\max\{m,n\})$. The sizes of $\mc{C}_1^*$ and $\mc{C}_2^*$ are both $O(m)$, therefore the size of $\mc{C}$ is $O(m)$. Therefore we prove (I).
	\end{proof}
	\begin{lemma}\label{lem62}
		Suppose $G=(V,E)$ is a graph, and $G_s$ is the underlying simple graph of $G$. Then $G$ admits a planar embedding if and only if $G_s$ does.
	\end{lemma}
	\begin{proof}
		It is obvious that if $G$ admits a planar embedding, then $G_s$ does. Now suppose $G_s$ admits a planar embedding. After we draw $G_s$ on the plane, there will always be room for parallel edges and self-loops. Therefore $G$ admits a planar embedding.
	\end{proof}
	\begin{proof}[\it{proof of Theorem \ref{thm11}}]
		The basic ideas are still the same as in the proof of Theorem \ref{thm52}, but we will design a different algorithm that exploits the planarity of $G/P_k(G)$ when solving the set symmetric difference problem. \par 
		$G/ P_k(G)$ is planar. By Theorem \ref{thm61}, we can construct a Eulerian-spanning set $\mc{C}$ in time $O((k+1)^2\max\{m,n\})$, such that $|\mc{C}|=O(m)$, and the number of edges that appear in at least three elements in $\mc{C}$ is at most $k$. Fix a planar embedding of $G/ P_k(G)$. We inherit the following notations from the proof of Theorem \ref{thm61}: \par 
		Subgraph $H=(V,P_k(G))$. Vertex sets $V_1^*,V_2^*,\cdots,V_t^*$. Graph $P_i\triangleq G[V_i^*]=(V_i^*,E_i^*)$. Spanning tree $T_i$ of $P_i$ such that $T_i\subseteq P_k(G)$. Edge set $E_T=\cup_{i=1}^tT_i$. The cycle $C_e\subseteq T_i\cup e$ for any edge $e\in E_i^*\rs T_i$. Set $\mc{C}_1^*$. Graph $G_p$ with a planar embedding computed and fixed as in the proof of Theorem \ref{thm61}. Eulerian-spanning set $\mc{C}_p=\{C_1,\cdots,C_l\}$ of $G_p$. Constructed edge set $C_i^*\in C(G)$. Set $\mc{C}_2^*$. \par
		All inherited notations carry the meaning exactly the same as in the proof of Theorem \ref{thm61}. \par
		Now we proceed. By Theorem \ref{thm33}, we only need to find an edge set $E_2$ with minimum total weight such that $\delta_1(E_2)=\mc{C}_o$, where $\mc{C}_o$ is the set of odd elements in $\mc{C}$. Let 
		$E_c^*=\{e\in E\rs E_T:\text{the two endpoints of $e$ are in the same connected component of }H\}$,
		$E_p^*=\{e\in E:\text{the two endpoints of $e$ are in different connected components of }H\}$.
		Then $E$ is the disjoint union of $E_T,E_c^*,E_p^*$, and there is a one-to-one correspondence between edges of $E_p^*$ and edges of $G_p$. Our algorithm is similar to that in the proof of Theorem \ref{thm52}. For each $E_3\subseteq E_T$, we find $E_{c}'\subseteq E_c^*,E_{p}'\subseteq E_p^*$ such that $\delta_1(E_3\cup E_c'\cup E_p')=\mc{C}_o$, and $E_c'\cup E_p'$ is of minimum total weight, or prove that there is no such $E_c',E_p'$. \par
		Let $\mc{C}_o'=\mc{C}_o\triangle(\delta_1(E_3)),\mc{C}_i^o=\mc{C}_o'\cap\mc{C}_i^*,i=1,2$. Since $\forall e\in E_c^*,\delta_1(\{e\})\subseteq\mc{C}_1^*,\forall e\in E_p^*,\delta_1(\{e\})\subseteq\mc{C}_2^*,\mc{C}_1^*\cap\mc{C}_2^*=\emptyset$, the problem reduces to finding $E_{c}'\subseteq E_c^*,E_{p}'\subseteq E_p^*$ such that $\delta_1(E_c')=\mc{C}_o^1,\delta_1(E_p')=\mc{C}_o^2$, and $E_c',E_p'$ are of minimum total weight, respectively. But $\forall e\in E_c^*,|\delta_1(\{e\})|=\{C_e\}$, and the cycles $C_e$ are pairwise different. Therefore $E_c'$ is simply the set of edges corresponding to elements in $\mc{C}_1^o$, and can be found in time $O(m)$. \par
		The only problem that remains is to find $E_p'\subseteq E_p^*$ of minimum total weight such that $\delta_1(E_p')=\mc{C}_o^2$, or claim that no such $E_p'$ exists. Let $G_p=(V_p,E_p)$. Recall that $\mc{C}_p$ is the set of boundaries of all faces of $G_p$, and thus is a Eulerian-spanning set of $G_p$. Recall that there is a one-to-one correspondence between edges in $E_p^*$ and edges in $E_p$, and a one-to-one correspondence between elements in $\mc{C}_2^*$ and $\mc{C}_p$. Let $f_1:E_p^*\to E_p,g_1:\mc{C}_2^*\to\mc{C}_p$ represent the correspondences. \par
		Let the dual graph of $G_p$ (recall that its planar embedding is already fixed) be $G_{d}=(V_d,E_d)$. Then elements in $V_d,E_d$ are in one-to-one correspondence with elements in $\mc{C}_p,E_p$, respectively. Let $f_2:E_p\to E_d,g_2:\mc{C}_p\to V_d$ represent the correspondences. Let $\delta_1'$ be the operator defined on $G_p$ with respect to the Eulerian-spanning set $\mc{C}_p$. Let $\partial_1''$ be the operator defined on $G_d$ (note that we do not need a Eulerian-spanning set to define this operator). Then we have $\forall E_p'\subseteq E_p^*,\partial_1''(f_2(f_1(E_p')))=g_2(\delta_1'(f_1(E_p')))=g_2(g_1(\delta_1(E_p')))$. Therefore, if we let $V_t=g_2(g_1(\mc{C}_2^o))$, and weight each edge $e\in E_d$ by the weight of $f_1^{-1}(f_2^{-1}(e))$, it suffices to find an edge set $E_d'\subseteq E_d$ with minimum total weight such that $\partial_1''(E_d')=V_t$, or claim that no such $E_d'$ exists. First, we note that self-loops in $E_d$ can all be removed since the image of a self-loop under $\partial_1''$ is zero. Therefore we assume that $G_d$ does not contain self-loops. Because the images of multiple edges with the same endpoints under $\partial_1''$ are identical, and all weights are non-negative, we can assume that at most one of these multiple edges are in $E_d'$. Therefore, we can further merge these multiple edges into a single one, and weight the single edge by the minimum weight among the multiple edges it comes from. Therefore, we assume that $G_d$ is a simple, connected planar graph from now on.\par
		Now we describe how to find $E_d'$ in time $O(\max\{m,n\}^{3/2}\log(\max\{m,n\}))$, or claim that no such $E_d'$ exists. If $2\nmid|V_t|$, then there can be no such $E_d'$ because any graph must have an even number of odd-degree vertices. Otherwise, we can convert the problem to a maximum weighted Eulerian graph problem, and then convert it to a maximum weighted perfect matching problem with a procedure almost identical to that in \cite{113}. We defer the conversion and the algorithm to Appendix \ref{a2} to avoid redundancy. We only need the conclusion that the procedure of converting the problem, solving it, and recovering the original answer costs time $O(\max\{m,n\}^{3/2}\log(\max\{m,n\}))$. \par
		Now we analyze the computational complexity of the above whole procedure. Computing the Eulerian-spanning set $\mc{C}$ according to Theorem \ref{thm61} costs time $O((k+1)^2\max\{m,n\})$. After we fixed $E_3\subseteq E_T$, converting the remaining problem, solving it, and recovering the original answer costs time $O(\max\{m,n\}^{3/2}\log(\max\{m,n\}))$. Therefore, the total computational complexity is $O(2^k\max\{m,n\}^{3/2}\log\max\{m,n\}+(k+1)^2\max\{m,n\})=O(2^k\max\{m,n\}^{3/2}\log\max\{m,n\})$.
	\end{proof}
	\begin{corollary}\label{col62}
		Suppose $G=(V,E)$ is a simple graph, $w:E\to \R^+\cup\{0\}$ is a weight function. Moreover, $E_1$ is another edge set included in the input such that $|E_1|=k$, and $G_{aug}\triangleq(V,E\cup E_1)$ is $k$-contraction apex with $P_k(G_{aug})=E_1$. Let $n=|V|,m=|E|$. Then there is an $O(2^k(\max\{m,n\}+k)^{3/2}\log(\max\{m,n\}+k))$-time algorithm finding the maximum weighted cut of $G$.
	\end{corollary} \par
	\begin{proof}
		We add the edges $E_1$ to $G$ except for those already in $G$, and weight the new edges by zero, getting the weighted graph $G_{aug}$ mentioned in the theorem. By Theorem \ref{thm11}, we can solve the weighted max-cut problem in $G_{aug}$ within time $O(2^k(\max\{m,n\}+k)^{3/2}\log(\max\{m,n\}+k))$. Since the maximum weighted cut in $G_{aug}$ is equivalent to the maximum weighted cut in $G$, we have our result.
	\end{proof} \par

	\section{Application to graphs with crossing number $k$}\label{s7}
	We first present the definition of $s$-planar graphs \cite{111} and crossing number of a given embedding.
	\begin{definition}
		A graph $G=(V,E)$ is $s$-planar if it can be embedded in the plane such that each edge is crossed at most $s$ times. The crossing number of an embedding is obtained in the following way: count the crossings separately for each pair of edges, and for each crossing point of the pair of edges. \par
		We say a $G$ has a $s$-planar embedding of crossing number $k$, if $G$ can be embedded in the plane such that each edge is crossed at most $s$ times, and the crossing number of this embedding is $k$. \par
		We define the crossing number of $G$ to be the minimum crossing number among all embeddings of $G$ in the plane.
	\end{definition}
	\begin{remark}
		Given an embedding of $G$ on the plane, if there are at least three edges that intersect in one point, we can adjust the edges a little in a local way such that each pair of edges intersect in a different point, without changing the total number of crossings. Therefore, we will always assume that $G$ is drawn on the plane so that situations in which three edges intersect in one point do not occur.
	\end{remark}
	\begin{lemma}\label{lem71}
		Suppose $G=(V,E)$ is a simple graph that has a 2-planar embedding of crossing number $k$. Then there is another edge set $E_1$ such that $|E_1|\leq k$, and if we let $G_{aug}=(V,E\cup E_1)$, then the underlying simple graph of $G_{aug}/E_1$, which we denote by $G_p$, is planar. Moreover, we can compute $E_1$ and the planar embedding of $G_p$ in time $O(km)$.
	\end{lemma}
	\begin{proof}
		We construct a procedure consisting of at most $k$ rounds. In each round, we do one of the following three things:\par
		i. Re-embed (or redraw) $G$ on the plane. \par
		ii. Contract an edge in $G$ while merging multiple edges and removing self-loops. \par
		iii. Pick an edge not in $G$, add it to $G$, and contract it while merging multiple edges and removing self-loops. \par
		We will show that after each round, $G$ is a simple graph that has a 2-planar embedding of smaller crossing number. Therefore, after at most $k$ rounds, we will get a planar graph, and the edges we have contracted during the procedure certainly form the set $E_1$ we wanted. \par
		Now suppose $G=(V,E)$ is a simple graph with a 2-planar embedding of crossing number $r$. Pick any intersection point $x$, and suppose it is the intersection of $e_1$ and $e_2$. We distinguish between three cases:\par
		(1) If $e_1$ and $e_2$ intersect at least twice, then by 2-planarity of the embedding, they must intersect exactly twice, and cannot possibly intersect any other edge. Therefore, $e_1,e_2$ can be redrawn in the way depicted in Figure \ref{f1}, such that $e_1$ and $e_2$ do not cross each other, and do not cross any other edges. Therefore the redrawing is indeed 2-planar, and the crossing number is reduced.\par
		(2) If $e_1$ and $e_2$ intersect exactly once at $x$, and share an endpoint, say, $v_1$. Let $e_1=(v_1,v_2),e_2=(v_1,v_3)$. By 2-planarity, the segments $v_1x(\text{on }e_1)$, $v_1x(\text{on }e_2)$, $xv_2,xv_3$ can intersect with other edges at most once each. Therefore, if we redraw the edges $e_1,e_2$ in the way depicted in Figure \ref{f1}, then they do not intersect each other, and can only intersect at most twice with other edges each. Therefore the redrawing is indeed 2-planar. Moreover, the crossing number is also reduced after the redrawing.
		\begin{figure}[h]
			\centering 
			\includegraphics[scale=.2]{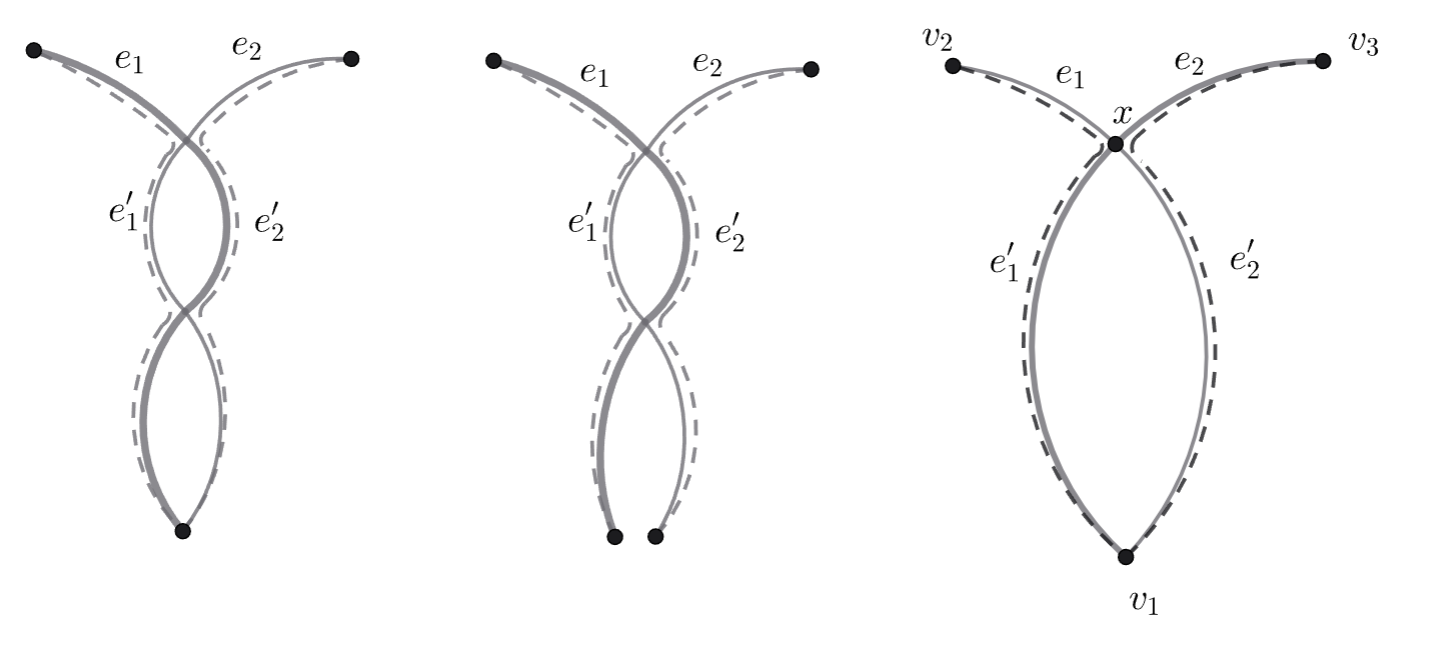}
			\caption{the case of two edges intersecting two times (figures in the left and middle, we distinguish cases whether they share an endpoint or not), and the case where two edges share an endpoint, but intersect only once (figure in the right). The original edges are those in solid lines, while the redrawn ones are those in dotted lines.}
			\label{f1}
		\end{figure} \par
		(3) If the endpoints of $e_1$ and $e_2$ are pairwise different, and they intersect only once on $x$. Let $e_1=(v_1,v_2),e_2=(v_3,v_4)$. By 2-planarity of the embedding, the segments $xv_1,xv_2$ on $e_1$ can only be crossed once by other edges altogether. Therefore there must be a segment, say $xv_1$, that is not crossed by any other edge. Similarly, let $xv_3$ be the segment on $e_2$ that is not crossed by any other edge. Let $e=(v_1,v_3)$. We further distinguish between two cases: \par
		(i) If $e\in E$, we first contract $e$ in $G$ while preserving multiple edges and self-loops, and redraw the edges $e_1,e_2$ and all edges incident to $v_1,v_3$ in the way depicted in Figure \ref{f2}. Specifically, we remove the edge $e$, move the vertices $v_1,v_3$ to $x$ along the segments $v_1x,v_3x$, and extend the edges incident to $v_1,v_3$ along the corresponding segments. Then we adjust the segments $v_1x,v_3x$ on edges incident to $v_1,v_3$ in a local way so that they won't cross each other, and won't cross other edges as well.\par
		Then an intersection point in the redrawing must also be an intersection point of $G$. Moreover, $x$ is no longer an intersection point of the redrawing. Therefore the redrawing is 2-planar, and of smaller crossing number. Then we merge multiple edges and remove self-loops, which cannot affect 2-planarity or increase the crossing number.\par
		(ii) If $e\notin E$, we add the edge $e$ to $G$ (we note that how we draw the edge $e$ on the plane is irrelevant, because we will contract it soon). Then we contract $e$ while preserving multiple edges and self-loops, and redraw the edges $e_1,e_2$ and all edges incident to $v_1,v_3$ in the way depicted in Figure \ref{f3}. Specifically, we remove the edge $e$, move the vertices $v_1,v_3$ to $x$ along the segments $v_1x,v_3x$, and extend the edges incident to $v_1,v_3$ along the corresponding segments. Then we adjust the segments $v_1x,v_3x$ on edges incident to $v_1,v_3$ in a local way so that they won't cross each other, and won't cross other edges as well.\par
		Then an intersection point in the redrawing must also be an intersection point of $G$. Moreover, $x$ is no longer an intersection point of the redrawing. Therefore the redrawing is indeed 2-planar, and of smaller crossing number. Then we merge multiple edges and remove self-loops, which cannot affect 2-planarity or increase the crossing number.\par
		\begin{figure}[h]
			\centering 
			\includegraphics[scale=.25]{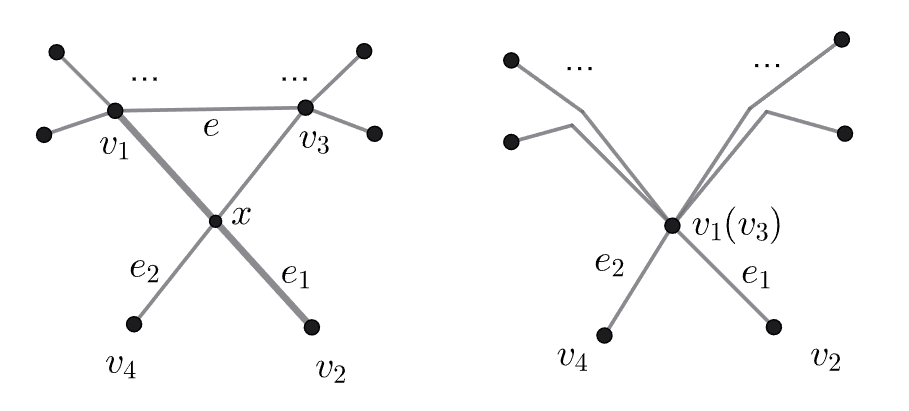}
			\caption{The case where the endpoints of $e_1$ and $e_2$ are pairwise different, and they intersect only once on $x$. If $e=(v_1,v_3)\in E$, we contract $e$ to the intersection point $x$, and redraw all edges incident to $v_1,v_3$ as depicted above (we have not merged multiple edges or removed self-loops so far).}
			\label{f2}
		\end{figure} \par
		\begin{figure}[h]
			\centering 
			\includegraphics[scale=.23]{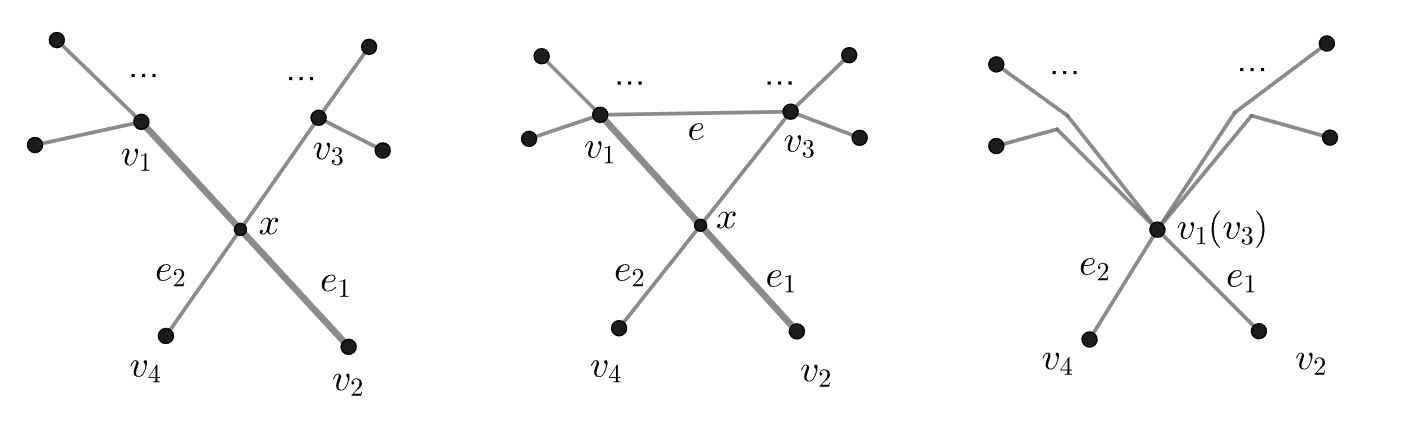}
			\caption{The case where the endpoints of $e_1$ and $e_2$ are pairwise different, and they intersect only once on $x$. If $e=(v_1,v_3)\notin E$, we first add $e$ to $G$, then contract $e$ to the intersection point $x$, and redraw all edges incident to $v_1,v_3$ as depicted above (we have not merged multiple edges or removed self-loops so far).}
			\label{f3}
		\end{figure} \par
		Now we consider the computational complexity of finding $E_1$ and planar embedding of $G_p$. Each round of our procedure can be done in time $O(m)$, since there are at most $m$ edges to be redrawn. Therefore the total running time of the procedure is $O(km)$.
	\end{proof}
	\begin{theorem}\label{thm71}
		Suppose $G=(V,E)$ is a simple graph with a 2-planar embedding of crossing number $k$. Let $n=|V|,m=|E|$. Then there is an $O(2^kn^{3/2}\log n)$-time algorithm finding the maximum weighted cut of $G$. 
	\end{theorem}
	\begin{proof}
		By \cite{111}, $m\leq 5n-10=O(n)$, therefore $k\leq 2m=O(n)$. Apply Lemma \ref{lem71}, Lemma \ref{lem62} and Corollary \ref{col62} we get our result.
	\end{proof} \par
	\begin{theorem}\label{thm72}
		Suppose $G=(V,E)$ is a simple graph with a planar embedding of crossing number $k$ included as input. Let $n=|V|,m=|E|$. Then there is an $O(2^k(n+k)^{3/2}\log (n+k))$-time algorithm finding the maximum weighted cut of $G$. 
	\end{theorem}
	\begin{proof}
		For any edge $e=(v_1,v_2)\in E$ do the following: let a point move from $v_1$ to $v_2$ through edge $e$, and the intersection points it encounters be $u_1,\cdots,u_l$ in chronological order. If $l<3$, we do nothing. Otherwise, for each $i\in\{1,2,\cdots,\lfloor\frac{l-1}{2}\rfloor\}$, add two vertices $z_{2i-1},z_{2i}$ between $u_{2i},u_{2i+1}$ on $e$. Weight every edge on the path $v_1z_1z_2\cdots z_{2\lfloor\frac{l-1}{2}\rfloor}v_2$ by weight of the edge $e$, which we denote by $w(e)$. Let the resulting graph be $G_2$. Then $G_2$ is a 2-planar graph with at most $n+2k$ vertices. Moreover, the maximum weighted cut problem of $G_2$ is equivalent to that of $G$ in the sense described as following: \par
		Consider the maximum weighted cut problem of $G_2$. If $v_1$ and $v_2$ belong to the same set of the partition, then at most $2\lfloor\frac{l-1}{2}\rfloor$ edges on the path $v_1z_1z_2\cdots z_{2\lfloor\frac{l-1}{2}\rfloor}v_2$ can be in the cut. Otherwise, there is a partition such that all $(2\lfloor\frac{l-1}{2}\rfloor+1)$ edges on the path $v_1z_1z_2\cdots z_{2\lfloor\frac{l-1}{2}\rfloor}v_2$ belong to the cut. Therefore, the effect of the path $v_1z_1z_2\cdots z_{2\lfloor\frac{l-1}{2}\rfloor}v_2$ is the same as a single edge $(v_1,v_2)$ weighted by $w(e)$. As a conclusion, the maximum weighted cut problem of $G_2$ is equivalent to that of $G$, and we can find the maximum weighted cut of $G$ in time $O(n)$ if we know the maximum weighted cut of $G_2$. \par
		By Theorem \ref{thm71}, the maximum weighted cut problem of $G_2$ can be solved in time $O(2^k(n+k)^{3/2}\log (n+k))$. Therefore the theorem is proven.
	\end{proof}

	\section{Conclusion and discussion}
	In this work, we generalize Hadlock's conversion of maxcut on planar graphs to one on general graphs with non-negative edge weights. We first convert the weighted maxcut problem to minimum weighted odd-circuit cover problem by means similar to Hadlock's. Then, by introducing the concept of Eulerian-spanning set and coboundary operator, we convert the problem of finding minimum weighted odd-circuit cover to finding a minimum weighted edge set with specific coboundary. By enumerating edges that appear at least three times in the Eulerian-spanning set, and solving the rest of the problem using Lemma \ref{lem51}, we finally solve the problem and recover the answer to maxcut. \par
	We obtain an algorithm finding the maximum cut of $k$-contraction apex graphs parameterized by $k$. We modified our algorithm to find a maximum cut on a graph with crossing number $k$. \par
	However, after we proved our result on $k$-contraction apex graphs, we realized that there is a direct algorithm finding the weighted maxcut within the same time limit. Specifically, we can enumerate the all possible partition within each connected component in $V,P_k(G)$ in time $O(2^k)$. After we fixed a partition within each connected component, the rest of the problem can be solved using the maxcut algorithm on planar graphs in time $O(\max\{m,n\}^{3/2}\log\max\{m,n\})$. Therefore the total running time is $O(2^k\max\{m,n\}^{3/2}\log\max\{m,n\})$. We haven't found results describing the above algorithm for $k$-contraction apex graphs yet. \par
	There are several possible directions of future work. First, a shortcoming of our result is that the set $P_k(G)$ is either included as input, or found efficiently in special cases, such as the graph being 2-planar. In general, we can only apply a brute force search to find the $P_k(G)$ if it is not revealed to us, which induces a heavy $\Omega(n^{k+1})$ running-time overhead. In future work, one may try to develop a more efficient algorithm finding the edge set $P_k(G)$ when it is not included in the input. Second, our conversion only covers the case of non-negative weights. One may try to generalize it to the setting of real weights. Third, one may try to directly generalize the planar duality to the setting of general graphs.

		\section*{Acknowledgements}
	This work was funded by the National Key R \& D Program of China (No.~2022YFA1005102) and
the National Natural Science Foundation of China (Nos.~12325112, 12288101).

	\appendix
	
	\section{NP-completeness of WSSD}\label{a1}
	Unfortunately, solving the decision version of WSSD, i.e. Problem \ref{prb32}, turns out to be NP-complete, even for unweighted version, and even if we restrict the size of any $S_i$ to be either 2 or 3. We will prove the above results in this section. \par
	Throughout this section, we will use $S,\mc{S},T,w,r$ to denote an instance of Problem \ref{prb32}. If the weight function is omitted, we mean that the weight of any subset is 1. \par
	First, we notice that Problem \ref{prb32} is in NP since any $\mc{T}\subseteq\mc{S}$ satisfying $\triangle_{S_i\in \mc{T}}S_i=T$, and $|\mc{T}|\leq r$ is a proof. Therefore, we only need to prove the NP-hardness in the following theorems.
	\begin{theorem}\label{thm41}
		Problem \ref{prb32} under the restriction that $w\equiv 1$ is NP-complete.
	\end{theorem}
	\begin{proof}
		To see that the problem is NP-hard, we reduce the decision version of finding a maximum cut in a Hamiltonian graph (where the Hamiltonian cycle can be included in the input) to the problem in the theorem, where the former has been proven to be NP-complete. \par
		Suppose $G=(V,E)$ is a simple graph with Hamiltonian cycle $C_H$. We construct a Eulerian-spanning set $\mc{C}$ of $G$ as follows: \par
		First we include $C_H$ into $\mc{C}$. Then, for every edge $e\notin C_H$, take a cycle $C_e\subseteq C_H\cup e$ such that $e\in C_e$ and include it into $\mc{C}$. We need to prove $\mc{C}$ is indeed a Eulerian-spanning set.\par
		If the sum of elements in $\mc{C}_1\subseteq\mc{C}$ is the empty set, then none of $C_e$ can be in $\mc{C}_1$ since $e$ is an edge that appears only in $C_e$. Therefore only $C_H$ can be in $\mc{C}_1$, which indicates that it cannot be in $\mc{C}_1$. Then $\mc{C}_1=\emptyset$. \par
		Moreover, for any Eulerian subgraph $G_s=(V,E_s)$, let $E_s\rs C_H={e_1,e_2,\cdots,e_t}$. Then $E_r=E_s+\sum_{i=1}^{t}C_{e_i}$ is another Eulerian subgraph that uses only edges in $C_H$. Therefore either $E_r=0$, or $E_r=C_H$, both cases indicating that $E_s$ is in the span of $\mc{C}$. \par
		Now that we have a Eulerian-spanning set $C_H$, we can convert it to an instance of Problem \ref{prb31} as described in the proof of Theorem \ref{thm34}. Then $G$ has a cut of size at least $r$ if and only if there is an odd-circuit cover of size at most $|E|-r$. which is equivalent to the existence of $\mc{T}\subseteq\mc{S}$ such that $\triangle_{S_i\in\mc{T}}S_i=T$, and $|\mc{T}|\leq |E|-r$. Moreover, the reduction can be performed in polynomial time with respect to $|V|$. The proof ends.
	\end{proof}
	\begin{theorem}
		Problem \ref{prb32} under the restriction that $w\equiv1$, and $\forall S_i\in\mc{S},|S_i|\in\{2,3\}$ is NP-complete. 
	\end{theorem}
	\begin{proof}
		To see that the problem is NP-hard, we reduce the problem in Theorem \ref{thm41} to this problem. We only need to take care of sets of size 1 or at least 4. We break the reduction into three steps: \par
		Step 1. In step 1, we eliminate all subsets of size 1. Suppose $S,\mc{S},T,r$ is an instance of the problem, and there is a set $S_i\in\mc{S}$ of size 1, and $S_i=\{s\}$, we do some modifications to the problem. Let
		$$S'=S\rs\{s\}\cup\{s^{(1)},s^{(2)}\}$$ 
		$$\mc{S}'=\{S_1',\cdots,S_m\},\ 
		S_k'=\left\{
		\begin{aligned}
			&S_k,\ if\ s\notin S_k \\
			&S_k\rs\{s\}\cup\{s^{(1)},s^{(2)}\}\ if\ s\in S_k
		\end{aligned}
		\right.
		$$
		$$T'=\left\{
		\begin{aligned}
			&T,\ if\ s\notin T \\
			&T\rs\{s\}\cup\{s^{(1)},s^{(2)}\}\ if\ s\in T
		\end{aligned}
		\right.
		$$ \par
		Then $S',\mc{S}',T',r$ is another instance of the problem in Theorem \ref{thm41}, which is equivalent to $S,\mc{S},T,r$. Moreover, the number of sets of size 1 in $\mc{S}'$ is strictly smaller than that of $\mc{S}$. Therefore, by repeating the above procedure, we will finally get an instance $S^{(1)},\mc{S}^{(1)},T^{(1)},r$ such that all sets in $\mc{S}^{(1)}$ are of size at least 2. \par
		Step 2. We give every set in $\mc{S}^{(1)}$ a weight of $p=|S^{(1)}|+2$, and denote the weight function by $w$. Now we have an instance, i.e. $S^{(1)},\mc{S}^{(1)},T^{(1)},w,pr$, which is equivalent to $S^{(1)},\mc{S}^{(1)},T^{(1)},r$. By a slight abuse of notation, we will denote this instance by $S,\mc{S},T,w,pr$, where $w\equiv|S|+2,p=|S|+2$.\par 
		Step 3. Decompose all subsets of size at least 4. We actually decompose every set in $\mc{S}$ in this step. Suppose $S_i\in\mc{S},S_i=\{s_i\}_{i=1}^h,w(S_i)=p$. Because $h\leq p-2$, we can do the following modifications: let
		$$S'=S\cup\{z_1,z_2,\cdots,z_{p-1},z_p\}$$
		$$\mc{S}'=\mc{S}\rs\{S_i\}\cup\mc{A}$$
		$$\mc{A}=\mc{A}_1\cup\mc{A}_2\cup\mc{A}_3$$
		$$\mc{A}_1=\{\{s_1,z_1\},\{s_2,z_1,z_2\},\{s_3,z_2,z_3\},\cdots,\{s_h,z_{h-1},z_h\}\}$$
		$$\mc{A}_2=\{\{z_h,z_{h+1}\},\cdots,\{z_{p-3},z_{p-2}\}\},\quad (\text{if }h=p-2,\mc{A}_2=\emptyset)$$
		$$\mc{A}_3=\{\{z_{p-2},z_{p-1},z_p\},\{z_{p-1},z_p\}\}$$
		$$\forall j\neq i, w'(S_j)=w'(S_i)$$
		$$\forall A\in\mc{A},w(A)=1$$
		Then $S',\mc{S}',T,w',pr$ is another instance equivalent to $S,\mc{S},T,w,pr$ . By repeating the above procedure, for every $S_i\in\mc{S},S_i=\{s_i\}_{i=1}^h,w(S_i)=p$, it is decomposed to sets of size 2 or 3, and of weight 1. Therefore we finally get an instance $S^{(2)},\mc{S}^{(2)},T,pr$ such that every set in $\mc{S}^{(2)}$ is of size 2 or 3. The reduction is complete, and can be performed in time polynomial to the size of $|S|$ before the reduction.
	\end{proof}
	
	\section{Completing the proof of Theorem \ref{thm11}}\label{a2}
	Now we are given a simple connected planar graph $G_d=(V_d,E_d)$ with weighted edges, a vertex set $V_t\subseteq V_d$ such that $2\mid|V_t|$. Let $\partial_1''$ be the operator defined on the graph $G_d$. Our goal is to find the minimum weighted edge set $E_d'\subseteq E_d$ such that $\partial_1''(E_d)=V_t$. \par
	We first convert the problem to a maximum weighted Eulerian subgraph problem. Let $V_o\subseteq V_d$ be the set of vertices of odd degree in $G_d$, and $V_c=V_o\triangle V_t$. Since $G_d$ is connected, we can compute an edge set $E_c\subseteq E_d$ such that $\partial_1''(E_c)=V_c$ as follows: take any pairing of vertices in $V_c$, and compute a path connecting each pair. The symmetric difference between these paths is the $E_c$ we want. Let $W$ be the total weight of $E_d$. For every $e\in E_c$, we add a duplicate of $e$ to the graph $G_d$, and weight the duplicated edge by $W+1$. Denote the graph obtained by $G_{daug}=(V_d,E_{aug})$, then it suffices to find $E_d'\subseteq E_{aug}$ with minimum total weight such that $\partial_1(E_d')=V_t$. Moreover, $V_t$ is the set of odd vertices in $G_{daug}$. Therefore, $\forall E_d'\subseteq E_{aug},\partial_1(E_d')=V_t$, we have that $(V_d,E_{aug}\rs E_d')$ is a Eulerian subgraph, and vice versa. Therefore, it suffices to find a maximum weighted Eulerian subgraph of $G_{daug}$. We note that there are at most two edges between any pair of vertices in $G_{daug}$, and $|E_{aug}|\leq 2|E_d|$ because $E_c\subseteq E_d$. Moreover, $G_{daug}$ is still a planar graph by Lemma \ref{lem62}. \par
	Now we convert the maximum weighted Eulerian subgraph problem to a maximum weighted perfect matching problem, using a procedure almost identical to that in \cite{113}. Starting with the graph $G_{daug}$, the conversion consists of two steps:\par
	Step 1. For every vertex $v\in V_d$ of degree at least five in $G_{daug}$, do the following: let the edges incident to $v$ be $e_1,e_2,\cdots,e_r$, in cyclic order. Split $v$ into $q=\lfloor \frac{r-1}{2}\rfloor$ vertices, $u_1,\cdots,u_q$. Add a path $u_1u_2\cdots u_q$, and weight all edges on the path by zero. Attach $e_1,e_2,e_3$ to $u_1$. Attach $e_{2i},e_{2i+1}$ to $u_i$, for $i=2,3,\cdots,q-1$. Attach $e_{2q},e_{2q+1},\cdots,e_r$ to $u_q$. This step of modification is depicted in Figure \ref{f4}. \par
	\begin{figure}[h]
		\centering 
		\includegraphics[scale=1.]{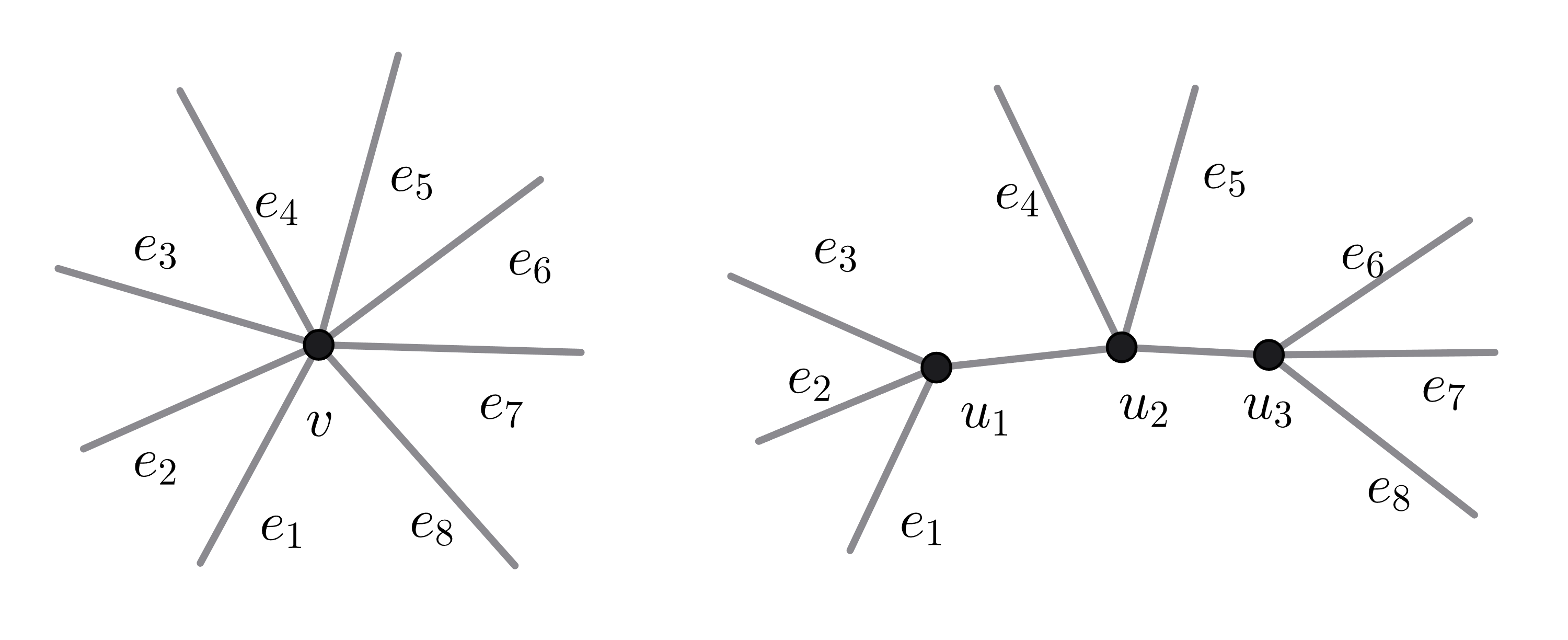}
		\caption{An example of modification in step 1. The vertex $v$ is of degree 8, and is split into three vertices.}
		\label{f4}
	\end{figure}
	Let the graph obtained after step 1 be $G_{d}^{(1)}=(V_d^{(1)},E_d^{(1)})$, which is still planar as is shown in Figure \ref{f4}. There are at most two edges between any pair of vertices in $G_d^{(1)}$ because $G_{daug}$ has the same property. Moreover, every vertex in $G_d^{(1)}$ is of degree at most 4. Every Eulerian subgraph of $G_d^{(1)}$ corresponds to a Eulerian subgraph of $G$ with the same total weight, and vice versa. Therefore we only need to find a maximum weighted Eulerian subgraph in $G_d^{(1)}$. We note that $|V_d^{(1)}|\leq |V_d|+|E_d|\leq\max\{m,n\}+m=O(\max\{m,n\})$,$|E_d^{(1)}|\leq|E_d|+|V_d^{(1)}|=O(\max\{m,n\})$.\par
	Step 2. For every vertex $v$ in $G_d^{(1)}$, do one of the following according to the degree of $v$:\par
	(i) If $v$ is of degree one, let the edge incident to $v$ be $e_1$. Split $v$ into $K_2$, weight the edge in $K_2$ by zero, and attach $e_1$ to a vertex in $K_2$. \par
	(ii) If $v$ is of degree two, let the edges incident to $v$ be $e_1,e_2$. Split $v$ into $K_2$, weight the edge in $K_2$ by zero, and attach $e_1,e_2$ to different vertices in $K_2$. \par
	(iii) If $v$ is of degree three, let the edges incident to $v$ be $e_1,e_2,e_3$. Split $v$ into $K_4$, weight the edges in $K_4$ by zero, and attach $e_1,e_2,e_3$ to three different vertices in $K_4$. \par
	(iv) If $v$ is of degree four, let the edges incident to $v$ be $e_1,e_2,e_3,e_4$. Split $v$ into $K_4$, weight the edges in $K_4$ by zero, and attach $e_1,e_2,e_3,e_4$ to four different vertices in $K_4$. \par
	The modification in this step is depicted in Figure \ref{f5}. \par
	\begin{figure}[h]
		\centering 
		\includegraphics[scale=.5]{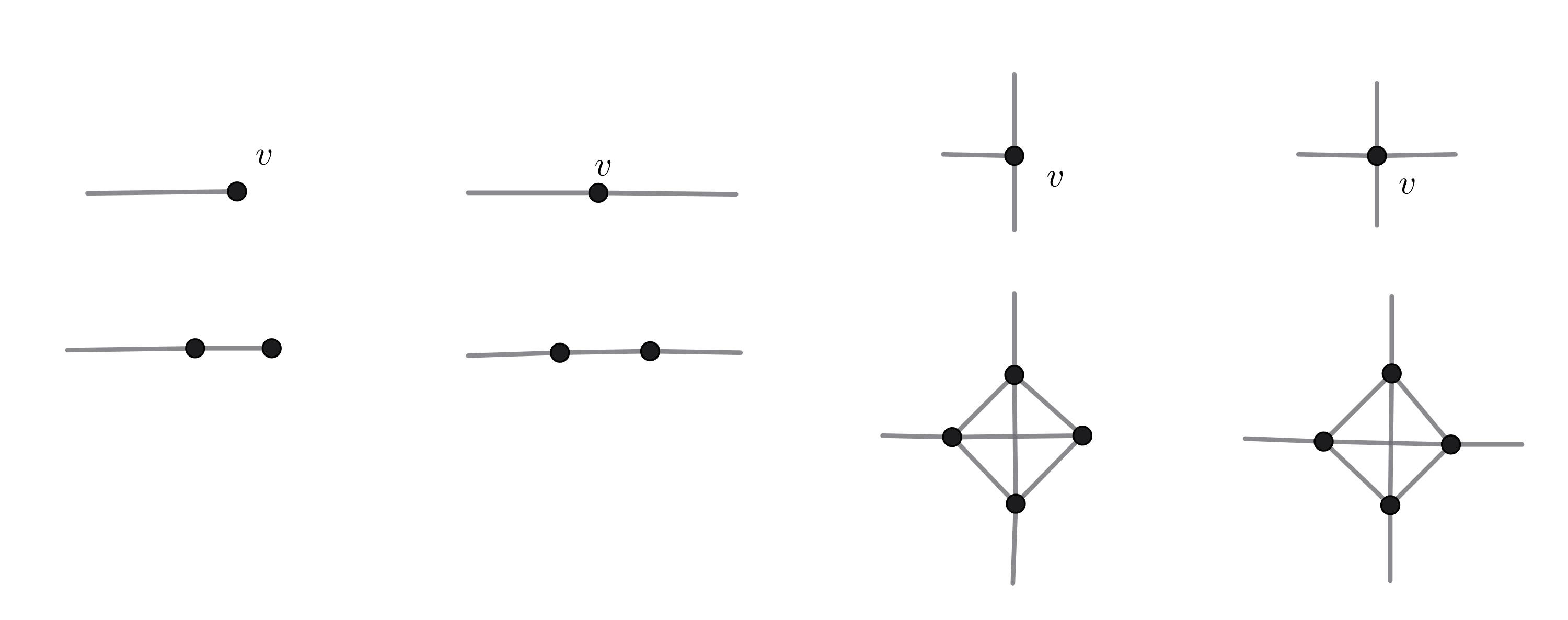}
		\caption{The modification in step 2. The figures, from left to right, correspond to the cases where $v$ is of degree 1,2,3,4, respectively.}
		\label{f5}
	\end{figure}
	Let $G_d^{(2)}=(V_d^{(2)},E_d^{(2)})$ be the graph obtained after this step. Then $G_d^{(2)}$ is simple because $G_d$ is assumed to be free of self-loops. It can be verified that, for every Eulerian subgraph of $G_d^{(1)}$, its corresponding edges in $G_d^{(2)}$ is a subset of a perfect matching with the same total weight. And for every perfect matching in $G_d^{(2)}$, the edges in $G_d^{(1)}$ corresponding to matching edges in $G_d^{(2)}$ form a Eulerian subgraph with the same total weight. The verifications are similar to those in \cite{113}, which already cover the cases where $v$ is of degree 3 and 4. Therefore we omit the verifications here. We note that $|V_d^{(2)}|\leq 4|V_d^{(1)}|=O(\max\{m,n\})$,$|E_d^{(2)}|\leq|E_d^{(1)}|+6|V_d^{(2)}|=O(\max\{m,n\})$. \par
	Now we only need to find a maximum weighted perfect matching in $G_d^{(2)}$. First, by adding a suitable large number to all weights, we can convert the problem to finding a maximum weighted matching (note that there exists a perfect matching in $G_d^{(2)}$, because empty graph is a Eulerian subgraph of $G_d^{(1)}$). Now we describe an algorithm similar to that in \cite{113,114}, using some results of this work:
	\begin{theorem}[\it{see Theorem 1 and Corollary 1 of \cite{114}}]\label{thm62}
		Let $G=(V,E)$ be any undirected simple planar graph of $n$ vertices. Let $n=|V|$. Then there is an $O(n)$-time algorithm finding a partition of $V$ into three sets $A,B,C$ satisfying the following: $|A|\leq2n/3,|B|\leq2n/3,|C|=O(\sqrt{n})$, and no edge of $E$ joins a vertex in $A$ and a vertex in $B$.
	\end{theorem} \par
	In the algorithm, we need to apply the above algorithm to the graph $G_d^{(1)}$, and its induced subgraphs. Although $G_d^{(1)}$ itself is not simple, there are at most two edges between any pair of vertices. Therefore the conclusion of Theorem \ref{thm62} still holds:
	\begin{corollary}\label{col61}
		Let $G=(V,E)$ be any undirected planar graph of $n$ vertices, such that there are at most two edges between any pair of vertices. Let $n=|V|$. Then there is an $O(n)$-time algorithm finding a partition of $V$ into three sets $A,B,C$ satisfying the following: $|A|\leq2n/3,|B|\leq2n/3,|C|=O(\sqrt{n})$, and no edge of $E$ joins a vertex in $A$ and a vertex in $B$.
	\end{corollary} \par
	\begin{theorem}[\it{see Lemma 2 in \cite{114} and the comments below it}]\label{thm63}
		Let $G=(V,E)$ be an undirected simple graph with edges weighted, and $v\in V$. Let $n=|V|,m=|E|$. Suppose we know the maximum weighted matching of $G-v$, then a maximum weighted matching of $G$ can be computed in time $O(m\log n)$.
	\end{theorem} \par
	Now we describe the algorithm, which begins with the simple graph $G_{d}^{(2)}$: \par
	Step 1. If $G=(V(G),E(G))$ contains at most one vertex, return the empty set as the maximum weighted matching. \par
	Step 2. Let the vertices of $G_d^{(1)}$ corresponding to vertices in $G$ form a set $V$. Apply the algorithm of Corollary \ref{col61} on the induced graph $G_d^{(1)}[V]$ to get a partition of $V$ into three sets $A,B,C$, where $|C|=O(\sqrt{|V|})$. Let $A_d,B_d,C_d\subseteq V(G)$ be the set of vertices corresponding to $A,B,C$. Then $|C_d|\leq4|C|=O(\sqrt{|V|})=O(\sqrt{|V(G)|})$. Apply our algorithm recursively to $G[A_d],G[B_d]$ to find the maximum weighted matching in $G[A_d],G[B_d]$. \par
	Step 3. Add vertices in $C_d$ to $G[A_d\cup B_d]$ one by one. Each time a vertex is added, apply Theorem \ref{thm63} to find a maximum weighted matching in the new graph. \par
	The description is completed. \par
	In Step 2, there are no edges of $G$ joining a vertex in $A_d$ and a vertex in $B_d$, therefore the union of maximum weighted matchings of $G[A_d],G[B_d]$ is indeed the maximum weighted matching of $G[A_d\cup B_d]$. Therefore the algorithm is correct. Let $t(n)$ be the running time of this algorithm on a graph with at most $n$ vertices. Then
	\begin{equation}
		\begin{aligned}
			& t(1)=c_1 \\
			\forall n>1,\ &t(n)\leq\max\{t(n_1)+t(n_2)+c_2n^{3/2}\log n:n_1+n_2=n,n_1,n_2\leq 2n/3\}
		\end{aligned}
	\end{equation} \par
	Where $c_1,c_2$ are positive constants. An inductive argument shows that $t(n)=O(n^{3/2}\log n)$. Therefore the time needed to compute a maximum weighted matching in $G_d^{(2)}$ is \\ $O(\max\{m,n\}^{3/2}\log(\max\{m,n\}))$. \par
	Moreover, converting the problem to a maximum weighted matching problem costs time at most $O(\max\{m,n\})$. Solving the maximum weighted matching problem costs time $O(\max\{m,n\}^{3/2}\log(\max\{m,n\}))$, and recovering the answer $E_p^*$ based on the solution also costs time at most $O(\max\{m,n\})$. Therefore the procedure in this appendix costs an overall $O(\max\{m,n\}^{3/2}\log(\max\{m,n\}))$ time.

\end{document}